\documentclass[final,twoside,11pt]{entics}
\usepackage{enticsmacro}
\usepackage{graphicx}
\usepackage[all]{xy}
\def\conf{MFPS 2024}
\def\myconf{mfps40} 	%%Fill in the acronym for your conference (with year)
\volume{4}			%Fill in the ENTICS volume number here
\def\volu{4}			% and here
\def\papno{8}			%Fill in your paper number here
\def\mydoi{14736}%
\def\lastname{Farjudian, Jung} %Lastnames appear in the running header
\def\runtitle{Continuous Domains for Function Spaces Using Spectral
    Compactification} %Short title appears in the running header on even pages. 
\def\copynames{A. Farjudian, A. Jung}    %% Fill in the first initial and last name of the authors
\usepackage[sumlimits]{amsmath}
\usepackage{amssymb}

\usepackage{diagrams}

\usepackage{enumitem}

\usepackage{tikz}
\usetikzlibrary{arrows}
\usetikzlibrary{fadings}

\usepackage{tikz-cd}
\usepackage{adjustbox}

\usepackage{todonotes}

\usepackage{mathtools}    % coloneqq

\newcommand{\aka}{also known as}
\newcommand{\eg}{e.\,g.}
\newcommand{\ie}{i.\,e.}

\newcommand{\cLat}[1]{\hat{\mathbb{C}}(#1)} % Lattice of closed
\newcommand{\defeq}{\coloneqq} % This is Amin's notation
\newcommand{\defiff}{\stackrel{\vartriangle}{\iff}} % definition by iff

\newcommand{\dirSubsetEq}{\ensuremath{\subseteq_{\mathit{dir}}}} % directed subset

\newcommand{\id}{\text{id}}

\newcommand{\interiorOf}[1]{\ensuremath{{#1}^{\circ}}}

\newcommand{\intvaldom}[1][{\Rinf}]{\ensuremath{{\mathbb{I}}#1}}

\newcommand{\binjoin}{\ensuremath{\vee}} % binary join (least upper bound in lattices)
\newcommand{\join}{\ensuremath{\bigvee}} % join (least upper bound in
\newcommand{\lift}{_\bot} % Adding a bottom element to a poset

\newcommand{\md}{\thinspace \mathrm{d}} % differential symbol

\newcommand{\binmeet}{\ensuremath{\wedge}} % binary meet (greatest lower bound in lattices)
\newcommand{\meet}{\ensuremath{\bigwedge}} % meet (greatest lower
\newcommand{\Q}{\ensuremath{\mathbb{Q}}}
\newcommand{\R}{\ensuremath{\mathbb{R}}}

\newcommand{\Scott}[1]{\ensuremath{\sigma_ {#1}}} % Scott topology of
\newcommand{\set}[1]{\ensuremath{\left\{{#1}\right\}}}
\newcommand{\setbarTall}[2]{\ensuremath{\left\{{#1}\;\vrule\;
{#2}\right\}}} % set comprehension notation with tall bar and braces
\newcommand{\setbarNormal}[2]{\ensuremath{\{{#1} \mid
{#2}\}}}% set comprehension notation with normal size bar and
\newcommand{\spect}[1]{\ensuremath{\hat{#1}}} % spectral
\newcommand{\State}{\mathbb{S}}

\newcommand{\wayaboves}[1]{\ensuremath{\twoheaduparrow {#1}}}
\newcommand{\waybelows}[1]{\ensuremath{\twoheaddownarrow {#1}}}

\DeclareMathOperator{\Idl}{\mathrm{Idl}} % Ideals (for ideal
\DeclareMathOperator{\lowerSet}{\downarrow} % lower set of an element
\DeclareMathOperator{\precAbsBas}{\vartriangleleft} % prec relation on abstract bases, usually written as \prec, which looks very similar to the usual < relation.

\DeclareMathOperator{\pt}{\mathrm{pt}} % The pt functor from lattices to topological spaces (in Stone duality)

\DeclareFontFamily{OT1}{pzc}{}
\DeclareFontShape{OT1}{pzc}{m}{it}{<->s*[1.30]pzcmi7t}{}
\DeclareMathAlphabet{\mathpzc}{OT1}{pzc}{m}{it}
\def\ct#1{\mathpzc{#1}} % font for (enriched) categories
\newcommand{\AFALCat}{\ct{Afal}} % Algebraic fully arithmetic lattices and frame homomorphisms

\newcommand{\BDLatCat}{\ct{BDLat}} % bounded distributive lattices and bounded lattice homomorphisms

\newcommand{\FrmCat}{\ct{Frm}} % Frames and frame homomorphisms

\newcommand{\opCat}[1]{\ensuremath{{#1}^{\mathit{op}}}} % Opposite category

\newcommand{\PoCat}{\ct{Po}} % Posets and monotonic maps

\newcommand{\TopCat}{\ct{Top}}%topological spaces and continuous maps

\newcommand{\SpecCat}{\ct{Spec}}% Spectral spaces 

\begin{document}

\begin{frontmatter}

  \title{Continuous Domains for Function Spaces Using Spectral
    Compactification}
  
 %  \title{An Example Paper\thanksref{ALL}} 						%%Title here and the
 % \thanks[ALL]{General thanks to everyone who should be thanked.}   %%Text of \thanks[ALL} here..
 %%%%%%%%%%%%%%%%%%%%%%%%%%%%			This Thanks is optional.
  %%%%Now the author(s) names(s)%%%%%
  \author{Amin Farjudian\thanksref{a}\thanksref{myemail}}	%%Note NO SPACE between 
   \author{Achim Jung\thanksref{b}\thanksref{coemail}}		%last name and \thanksref{...} 
    %%%Next come the addresses%%%%
   \address[a]{School of Mathematics\\University of
     Birmingham\\
     Birmingham, United Kingdom}  							
   \thanks[myemail]{Email: \href{mailto:A.Farjudian@bham.ac.uk} {\texttt{\normalshape
        A.Farjudian@bham.ac.uk}}} 
   %%%Note: if both authors share same institution, only list the address once, after the second 
   %%%author. 
   %%%There also is a link from the first author to the co-author's address to show how to list 
   %%%affiliations to more than one institution, when needed. 
  \address[b]{School of Computer Science\\University of
     Birmingham\\
     Birmingham, United Kingdom} 
  \thanks[coemail]{Email:  \href{mailto:A.Jung@bham.ac.uk} {\texttt{\normalshape
        A.Jung@bham.ac.uk}}}

  \begin{abstract}
    We introduce a continuous domain for function spaces over
    topological spaces which are not core-compact. Notable examples of
    such topological spaces include the real line with the upper limit
    topology, which is used in solution of initial value problems with
    temporal discretization, and various infinite dimensional Banach
    spaces which are ubiquitous in functional analysis and solution of
    partial differential equations. If a topological space
    $\mathbb{X}$ is not core-compact and $\mathbb{D}$ is a
    non-singleton bounded-complete domain, the function space
    $[\mathbb{X} \to \mathbb{D}]$ is not a continuous domain. To
    construct a continuous domain, we consider a spectral
    compactification $\mathbb{Y}$ of $\mathbb{X}$ and relate
    $[\mathbb{X} \to \mathbb{D}]$ with the continuous domain
    $[\mathbb{Y} \to \mathbb{D}]$ via a Galois
    connection. This allows us to perform computations in the native
    structure $[\mathbb{X} \to \mathbb{D}]$ while computable analysis
    is performed in the continuous domain
    $[\mathbb{Y} \to \mathbb{D}]$, with the left and right
    adjoints used for moving between the two function spaces.
\end{abstract}
\begin{keyword}
  domain theory, compactification, Stone duality
\end{keyword}

\end{frontmatter}

%%%%%%%%%%%%%%%%%%%%%%%%%%%%%%%%%%%%%%%%%%%%%%%%%%%%%%%%%
\section{Introduction}
\label{sec:intro}

The tight link between topology and the theory of computation is
well-known and has been investigated extensively in the
literature. This link is clearly manifested in the theory of
domains~\cite{Mislove-Topology_DT_TCS:1998}, which have, in
particular, provided a natural computational framework for
mathematical analysis. This line of research was initiated by Edalat's
work on dynamical systems~\cite{Edalat95:DT-fractals}. Ever since,
domains have been used for the study of several other concepts and
operators of mathematical analysis, {\eg}, exact real number
computation~\cite{Escardo96-tcs,Edalat:Domains_Physics:1997},
differential equation solving~\cite{Edalat_Pattinson2007-LMS_Picard},
stochastic processes~\cite{Bilokon_Edalat:Domain_Brownian:2017},
reachability analysis of hybrid
systems~\cite{Edalat_Pattinson:Hybrid:2007,Moggi_Farjudian_Duracz_Taha:Reachability_Hybrid:2018},
and robustness analysis of neural
networks~\cite{Zhou_Shaikh_Li_Farjudian:Robust_NN:MSCS:2023}.

In such applications, when the topological spaces involved have some
desirable properties ({\eg}, metrizability, local compactness, etc.)
the construction of the domain model can be relatively
straightforward. In the absence of favourable properties, however,
domain models do not arise naturally and one may look for
\emph{substitute} constructions, an example of which can be found
in~\cite{Farjudian_Moggi:Robustness_Scott_Continuity_Computability:2023}
for robustness analysis of systems with state spaces that are not
(locally) compact.

Another example is encountered in the solution of initial value
problems (IVPs). For the Picard method of IVP solving, continuous
domains of functions arise
naturally~\cite{Edalat_Lieutier:Domain_Calculus_One_Var:MSCS:2004,Edalat_Pattinson2007-LMS_Picard,FarjudianKonecny2008:wollic-lnai}. The
situation is slightly different for the methods that are based on
temporal discretization ({\eg}, Euler and Runge-Kutta methods). While
it is still possible to use classical domain models when an imperative
style of computation is
adopted~\cite{Edalat_Farjudian_Mohammadian_Pattinson:2nd_Order_Euler:2020:Conf},
a functional implementation via the fixpoint operator requires a
substitute domain
construction~\cite{Edalat_Farjudian_Li:Temporal_Discretization:2023}. Let
us discuss this in more detail. Assume that the following IVP is
given:
\begin{equation}
  \label{eq:main_ivp}
    \left\{
      \begin{array}{r@{\hspace{0.5ex}=\hspace{0.5ex}}l}
        y'(t) & f(y(t)),\\
        y(t_0) & y_0,\\
      \end{array}
    \right.
  \end{equation}
  in which $t_0 \in \R$, $y_0 \in \R^n$, and $f: \R^n \to \R^n$ is a
  continuous vector field, for some natural number $n \geq 1$. Assume
  that a solution exists over a lifetime of $[t_0,T]$, for some
  $T > t_0$. In a domain-theoretic framework, one would search for a
  solution of~\eqref{eq:main_ivp} in the space of functions from
  $[t_0, T]$ to the interval domain:
  \begin{equation}
    \label{eq:interval_domain_IR_n}
    \intvaldom[\R^n\lift] \defeq \set{\R^n} \cup \setbarTall{\prod_{i=1}^n [
      a_i, b_i]}{\forall i \in \set{1, \ldots, n} : a_i, b_i \in \R
      \text{ and } a_i \leq b_i},
  \end{equation}
  ordered by superset relation, {\ie},
  $\forall X, Y \in \intvaldom[\R^n\lift]: X \sqsubseteq Y \defiff X
  \supseteq Y $. Hence, the set $\R^n$ is the least element of the
  interval domain $\intvaldom[\R^n\lift]$.

  For applications such as differential equation solving, the interval
  domain $\intvaldom[\R^n\lift]$ is considered with the Scott
  topology. What is equally important is the topology on the interval
  $[t_0, T]$. For the Picard method, the Euclidean topology on
  $[t_0, T]$ is the suitable topology. As the Euclidean topology over
  $[t_0,T]$ is locally compact, the space of functions from $[t_0,T]$
  to $\intvaldom[\R^n\lift]$---which are continuous with respect to
  the Euclidean topology on $[t_0, T]$ and the Scott topology on
  $\intvaldom[\R^n\lift]$---ordered by pointwise ordering is a
  continuous domain.

The Euclidean topology, however, is not suitable in a functional
framework in the presence of temporal discretization. To see this,
note that, by integrating both sides of~\eqref{eq:main_ivp}, we obtain
$y( t+h) = y(t) + \int_t^{t+h} f(y(\tau)) \md \tau$, for all
$t \in [t_0,T]$ and $h \in [0,T-t]$. This can be written as:
\begin{equation}
  \label{eq:y_t_h_i}
  y(t+h) =
  y(t) + i( t, h),
\end{equation}
in which the integral $i( t, h)$ represents the dynamics of the
solution from $t$ to $t+h$. Thus, a general schema for validated
solution of the IVP~\eqref{eq:main_ivp} with temporal discretization
may be envisaged as follows:

\begin{enumerate}[label=(\roman*)]

\item For some $k \geq 1$, consider the partition
  $Q=(q_0, \ldots, q_k)$ of the interval $[t_0,T]$.

\item Let $Y(t_0) \defeq y_0$.

\item \label{item:gen_schema_iteration} For each
  $j \in \set{0, \ldots, k-1}$ and $h \in (0,q_{j+1} - q_j]$:
  \begin{equation}
    \label{eq:Y_h_I_j}
    Y(q_j+h) \defeq Y(q_j) + I( q_j, h),
  \end{equation}
  where $I( q_j, h)$ is an interval enclosure of the integral factor
  $i( q_j, h)$ from equation~\eqref{eq:y_t_h_i}. The operator $I$, in
  general, depends on several parameters, including (enclosures of)
  the vector field and its derivatives, the enclosure $Y(q_j)$, the
  index $j$, etc.
  
\end{enumerate}
In~\eqref{eq:Y_h_I_j}, the operator `$+$' denotes interval addition,
and for the method to be validated, the term $I( q_j, h)$ must account
for all the inaccuracies, {\eg}, floating-point error, truncation
error, etc.

In step~\ref{item:gen_schema_iteration} of the schema, the solver
moves forward in time, from $q_j$ to $q_{j+1}$. This requires keeping
the state, {\ie}, the solution up to the partition point $q_j$, and
referring to this state in iteration~$j$. As such, the schema has an
imperative style, and indeed encompasses various validated approaches
to IVP solving with temporal discretization in the literature,
including~\cite{Edalat_Farjudian_Mohammadian_Pattinson:2nd_Order_Euler:2020:Conf}. This
is in contrast with the functional style adopted in the definition of
the Picard operator
in~\cite{Edalat_Lieutier:Domain_Calculus_One_Var:MSCS:2004,Edalat_Pattinson2007-LMS_Picard},
and in language design for real number computation. For instance, the
languages designed
in~\cite{Escardo96-tcs,Farjudian:Shrad:2007,DiGianantonio_Edalat_Gutin-Automatic_Differentiation:2023}
for computation over real numbers and real functions are functional
languages based on lambda calculus, with their denotational semantics
provided by domain models.

In a functional framework, the solution of the IVP~\eqref{eq:main_ivp}
is obtained as the fixpoint of a higher-order operator. Domain models
are particularly suitable for fixpoint computations of this type. A
straightfoward (but, flawed) way of obtaining a fixpoint formulation
for the above general schema is to define a functional $\Phi$ over
interval functions as follows:
\begin{equation*}
  %\label{eq:Phi_fixpoint_flawed}
  \Phi(Y)(x) \defeq \left\{
      %\arrayoptions{2ex}{1.3}
      \begin{array}{ll}
        y_0, & \text{if } x = t_0,\\
        Y( q_j) + I( q_j, x-q_j), & \text{if } q_j < x \leq q_{j+1}.
      \end{array}
      \right.
\end{equation*}

\noindent
The fixpoint of this operator (if it exists) will be the right
choice. The problem is that, the enclosures obtained by applying
$\Phi$ do not have upper (respectively, lower) semi-continuous upper
(respectively, lower) bounds, even if the initial enclosure has
continuous bounds. This situation is illustrated in
Fig.~\ref{fig:fixpoint_flawed}, where the true solution of the IVP is
drawn in {\color{magenta}{magenta}} color. We begin with an initial
enclosure of the solution drawn in {\color{green}{green}} color,
{\ie}, the outer solid curves. As can be seen, the upper and lower
bounds of this initial enclosure are continuous.

After applying $\Phi$ once, we obtain the first approximation of the
solution, with piecewise affine upper (in {\color{black}{solid
    black}}) and lower (in {\color{blue}{solid blue}}) bounds. As can
be seen at points $q_0$, $q_1$, $q_2$, and $q_3$, the upper bound is
not upper semi-continuous, and the lower bound is not lower
semi-continuous. If we apply the operator $\Phi$ a second time, we
obtain a tighter ({\ie}, more accurate) enclosure of the solution,
which is drawn in dashed lines. But these bounds also have the same
problem with semi-continuity, {\ie}, the upper bound is not upper
semi-continuous, and the lower bound is not lower
semi-continuous. Hence,
by~\cite[Proposition~2.10]{Edalat_Farjudian_Li:Temporal_Discretization:2023},
this results in approximations of the solution of the IVP which are
not continuous with respect to the Euclidean topology on $[t_0, T]$,
and can only be continuous with respect to the so-called upper limit
topology. Recall that the upper limit topology has as its base the
collection $\setbarNormal{(a,b]}{a,b \in \R}$ of half-open
intervals. This topology is known not to be locally compact (see,
{\eg},~\cite[Proposition~4.5]{Edalat_Farjudian_Li:Temporal_Discretization:2023}). This
shortcoming motivated the substitute construction presented
in~\cite{Edalat_Farjudian_Li:Temporal_Discretization:2023}, where a
more detailed justification of why the upper limit topology is needed
can also be found, together with fixpoint formulations of Euler and
Runge-Kutta operators.

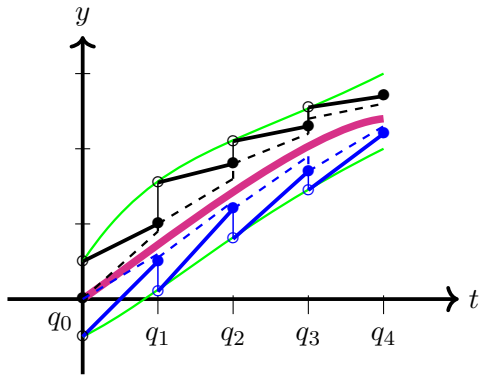
\begin{figure}[h]
  \centering

    \adjustbox{scale=1, center}{
      \begin{tikzpicture}
	
	% Axes
	\draw [->, line width = 0.5mm] (-1,0) -- (5,0);
	\draw [->, line width = 0.5mm] (0,-1) -- (0,3.5);
	\node [right] at (5,0) {$t$};
	\node [above] at (0,3.5) {$y$};

        % Ticks on the axes
        
        \draw (1,-0.2) -- (1,0.05);
        \draw (2,-0.2) -- (2,0.05);
        \draw (3,-0.2) -- (3,0.05);
        \draw (4,-0.2) -- (4,0.05);

        \draw (-0.1,1) -- (0.1,1);
        \draw (-0.1,2) -- (0.1,2);
        \draw (-0.1,3) -- (0.1,3);

        % axes marks
        \node at (-0.3, -0.3) {$q_0$};
        \node at (1, -0.5) {$q_1$};
        \node at (2, -0.5) {$q_2$};
        \node at (3, -0.5) {$q_3$};
        \node at (4, -0.5) {$q_4$};

        % (Hypothetical) solution
        %\draw[color=magenta!50!black, line width=1mm] (0,0)
        %.. controls (1,0.7) and ( 3,2.3) .. (4,2.4);
        \draw[color=magenta!90!black, line width=1mm] (0,0) .. controls (1,0.7) and ( 3,2.3) .. (4,2.4);

        % The initial enclosure of the solution
        \draw[color=green, line width = 0.3mm] (0,0.5) .. controls (1,2) and (2,2) .. (4,3);
        \draw[color=green, line width = 0.3mm] (0,-0.5) .. controls (1,0) and (2,1) .. (4,2);

        % Enclosures from the fixpoint operator
        
        \node at ( 0, 0) {\textbullet};

        % The upper bound, which is not upper semicontinuous
        % and The lower bound, which is not lower semicontinuous

        \node at ( 0, 0.5) {$\circ$}; 
        \draw[line width = 0.5mm] (0,0.5) -- (1,1); \node at ( 1, 1)
        {\textbullet};

        \draw[color=blue, line width = 0.5mm] (0,-0.5) -- (1,0.5); \node[color=blue] at ( 1, 0.5)
        {\textbullet};

        \node at ( 0, -0.5) {$\circ$};

        \draw[color=blue, line width = 0.2mm] (1, 0.5) -- (1, 0.1);

        \draw[color=black, line width = 0.2mm] (1, 1)
        -- (1, 1.55);

        \node at ( 1, 1.55) {$\circ$}; \draw[line width = 0.5mm]
        (1,1.55) -- (2,1.8); \node at ( 2, 1.8) {\textbullet};

        \node[color=blue] at ( 1, 0.1) {$\circ$}; \draw[color=blue,
        line width = 0.5mm] (1, 0.1) -- (2, 1.2); \node[color=blue] at ( 2, 1.2)
        {\textbullet};

        \draw[color=black,  line width = 0.2mm] (2, 1.8)
        -- (2, 2.1);

        \draw[color=blue, line width = 0.2mm] (2, 1.2) -- (2, 0.8);

        \node at ( 2, 2.1) {$\circ$}; \draw[line width = 0.5mm]
        (2,2.1) -- (3,2.3); \node at ( 3, 2.3) {\textbullet};

        \node[color=blue] at ( 2, 0.8) {$\circ$}; \draw[color=blue, line width = 0.5mm]
        (2, 0.8) -- (3,1.7); \node[color=blue] at ( 3, 1.7) {\textbullet};

        \draw[color=blue, line width = 0.2mm] (3, 1.7)
        -- (3, 1.45);

        \draw[color=black,  line width = 0.3mm] (3,
        2.3) -- (3, 2.55);

        \node at ( 3, 2.55) {$\circ$}; \draw[line width = 0.5mm]
        (3,2.55) -- (4,2.7); \node at ( 4, 2.7) {\textbullet};

        \node[color=blue] at ( 3, 1.45) {$\circ$}; \draw[color=blue,
        line width = 0.5mm] (3, 1.45) -- (4,2.2); \node[color=blue] at
        ( 4, 2.2) {\textbullet};

        \draw[color=black, dashed, line width = 0.3mm] (0,0) -- (1,0.9);
        \draw[color=blue, dashed, line width = 0.3mm] (0,0) -- (1,0.6);

        \draw[color=black, dashed, line width = 0.3mm] (1,0.9) -- (1, 1);
        \draw[color=black, dashed, line width = 0.3mm] (1,1) -- (2, 1.6);
        \draw[color=blue, dashed, line width = 0.3mm] (1,0.55) -- (2,1.3);

        \draw[color=black, dashed, line width = 0.3mm] (2,1.6) -- (2,1.8);
        \draw[color=black, dashed, line width = 0.3mm] (2,1.8) -- (3, 2.2);
        \draw[color=blue, dashed, line width = 0.3mm] (2,1.2) -- (3,1.9);

        \draw[color=black, dashed, line width = 0.3mm] (3,2.2) -- (3,2.4);
        \draw[color=black, dashed, line width = 0.3mm] (3,2.4) -- (4, 2.6);

        \draw[color=blue, dashed, line width = 0.3mm] (3,1.9) -- (3,1.7);
        \draw[color=blue, dashed, line width = 0.3mm] (3,1.7) -- (4,2.3);

      \end{tikzpicture}
    }

    \caption{The correct semi-continuity of the bounds of the
      successive approximations cannot be guaranteed, even when the
      process starts from an initial enclosure with continuous bounds
      (in {\color{green}{green}}). The curve in
      {\color{magenta}{magenta}} depicts the true solution, and the
      {\color{black}{black}} (respectively, {\color{blue}{blue}})
      piecewise affine curves represent the upper (respectively,
      lower) bounds of the enclosures obtained from applying $\Phi$
      once (solid) and twice (dashed).}
  \label{fig:fixpoint_flawed}
\end{figure}

In a more general setting, assume that
$\mathbb{X} \equiv (X, \tau_{\mathbb{X}})$ is a $T_0$ topological
space, and $\mathbb{D}$ is a non-singleton bounded-complete domain
(bc-domain). From~\cite{Erker_et_al:way_below:1998}, we know that the
function space $[\mathbb{X} \to \mathbb{D}]$ is a bc-domain if and
only if $\mathbb{X}$ is core-compact. Thus, the main challenge is to
develop a computational framework for the function space
$[\mathbb{X} \to \mathbb{D}]$ when $\mathbb{X}$ is not core-compact,
{\eg}, when $\mathbb{X}$ is the interval $[t_0,T]$ endowed with the
upper limit topology, and $\mathbb{D} = \intvaldom[\R^n\lift]$ with
the Scott
topology. In~\cite{Edalat_Farjudian_Li:Temporal_Discretization:2023},
a domain is constructed as a substitute for this function space via
abstract bases. In contrast
with~\cite{Edalat_Farjudian_Li:Temporal_Discretization:2023}, the aim
of the current article is to show that we can work directly on the
space $\mathbb{X}$ and obtain the same substitute via Stone
duality. To be more specific, using the well-known results in Stone
duality, we construct a topological space $\spect{\mathbb{X}}$ with
the following properties:

  \begin{itemize}
  \item $\spect{\mathbb{X}}$ is a core-compact (in fact, spectral)
    space and $\mathbb{X}$ can be embedded into $\spect{\mathbb{X}}$ as
    a dense subspace.
  \item The function spaces $[\mathbb{X} \to \mathbb{D}]$ and
    $[\spect{\mathbb{X}} \to \mathbb{D}]$ are related via a Galois connection.
  \end{itemize}

  Such a construction is useful for computable analysis within a
  domain framework. When $\mathbb{X}$ is not core-compact, the
  non-continuous directed-complete partial order (dcpo)
  $[\mathbb{X} \to \mathbb{D}]$ is used for implementation of
  algorithms, whereas analysis of computability is carried out over
  the continuous domain $[\spect{\mathbb{X}} \to \mathbb{D}]$, subject
  to the existence of a suitable effective structure over
  $[\spect{\mathbb{X}} \to \mathbb{D}]$.

%%%%%%%%%%%%%%%%%%%%%%%%%%%%%%%%%%%%%%%%%%%%%%%%%%%% 
  \subsection{Related Work}

  Compactification is a fundamental concept in topology. Classical
  examples such as Stone-{\v C}ech and one-point
  compactification~\cite{Munkres:Topology:2000} have been introduced
  primarily for Hausdorff topological spaces. In the non-Hausdorff
  setting, Smyth's stable
  compactification~\cite{Smyth:Stable_Compactification:1992} is the
  closest to ours. In fact, our construction can be obtained as a
  special case of Smyth's stable compactification by considering the
  so-called fine quasi-proximities, resulting in compactifications
  that are spectral. In \cite{Smyth:Stable_Compactification:1992},
  this special case is refered to as spectralization, whereas we use
  the term \emph{spectral compactification} to emphasize the
  compactification aspect of the construction. Spectral
  compactification is indeed an important special case of stable
  compactification which is suitable for computational purposes. Here,
  we keep the presentation simple and do not use quasi-proximities
  which form the foundation of Smyth's construction. We obtain all the
  basic properties that we need in this simpler framework. We point
  out that spectral compactification is fundamentally different from
  Stone-{\v C}ech and one-point compactifications in that, even when a
  space $\mathbb{X}$ is compact and Hausdorff, its spectral
  compactification may not be $T_1$ (for an example,
  see~\cite[page~338]{Smyth:Stable_Compactification:1992}).

  Another aspect of our work here is the idea of a substitute
  construction. Such constructions can be useful when the topological
  spaces do not have favourable
  properties. In~\cite{Edalat_Farjudian_Li:Temporal_Discretization:2023},
  we used the idea of a substitute construction in the context of IVP
  solving. Another example is presented
  in~\cite{Farjudian_Moggi:Robustness_Scott_Continuity_Computability:2023},
  in the context of robustness
  analysis. In~\cite{Farjudian_Moggi:Robustness_Scott_Continuity_Computability:2023},
  we studied robustness analysis of systems with state spaces $\State$
  which are not locally compact. In such cases, the lattice
  $\cLat{\State}$ of closed subsets of $\State$ (under superset
  relation) may not be continuous, let alone $\omega$-continuous. The
  lattice of closed subsets is central to robustness analysis.  Hence,
  we construct an $\omega$-continuous lattice $\mathbb{L}$ which is
  related to $\cLat{\State}$ via a suitable adjunction.

%%%%%%%%%%%%%%%%%%%%%%%%%%%%%%%%%%%%%%%%%%%%%%%%%%%%
\subsection{Structure of the Paper}
\label{subsec:structure}

The preliminaries, including a brief reminder of basic concepts from
domain theory and Stone duality, are presented in
Section~\ref{sec:prelim}. In
Section~\ref{sec:Basic_Galois_Connection}, we establish a Galois
connection between the function spaces $[\mathbb{X} \to \mathbb{D}]$
and $[\mathbb{Y} \to \mathbb{D}]$, where $\mathbb{D}$ is a bc-domain,
$\mathbb{X}$ and $\mathbb{Y}$ are topological spaces, and $\mathbb{X}$
is densely embedded in $\mathbb{Y}$. A detailed account of the
spectral compactification of topological spaces is presented in
Section~\ref{sec:Core_Compactification}. A continuous domain for the
space of functions from an arbitrary $T_0$ space $\mathbb{X}$ to a
bc-domain $\mathbb{D}$ is constructed in Section~\ref{sec:domain_funs}
based on spectral compactification of $\mathbb{X}$, and we prove that
the result is equivalent to the construction based on abstract bases
developed
in~\cite{Edalat_Farjudian_Li:Temporal_Discretization:2023}. We
conclude the article with some remarks in
Section~\ref{sec:Concluding_Remarks}.

%%%%%%%%%%%%%%%%%%%%%%%%%%%%%%%%%%%%%%%%%%%%%%%%%%%%%%%%%
\section{Preliminaries}
\label{sec:prelim}

Basic familiarity with domain theory and Stone
duality~\cite{AbramskyJung94-DT,Goubault-Larrecq:Non_Hausdorff_topology:2013}
will be helpful in understanding the main results of the paper.  In
this section, we present a brief reminder of the preliminary concepts,
and establish some notations and definitions.

For arbitrary sets $X$ and $Y$, by $X \subseteq_f Y$ we mean $X$ is a
finite subset of $Y$. Assume that $(D, \sqsubseteq)$ is a partially
ordered set (poset) and $A \subseteq D$. We define
$\lowerSet{A} \defeq \setbarNormal{x \in D}{\exists a \in A: x
  \sqsubseteq a}$, and when $A$ is a singleton $\set{a}$, we may
simply write $\lowerSet{a}$ instead of $\lowerSet{\set{a}}$. We denote
the join ({\aka} the least upper bound) of $A$ by $\join A$, and the
meet ({\aka} the greatest lower bound) of $A$ by $\meet A$, whenever
they exist. A subset $A \subseteq D$ is said to be directed if it is
non-empty and $\forall x,y \in A: \exists z \in A: x \sqsubseteq z$
and $y \sqsubseteq z$, in which case, we write $A \dirSubsetEq D$. The
poset $(D, \sqsubseteq)$ is said to be a directed-complete partial
order (dcpo) if $\forall A \dirSubsetEq D: \join A$ exists in $D$. The
poset $(D, \sqsubseteq)$ is said to be pointed if it has a bottom
element $\bot$.

Assume that $(D, \sqsubseteq)$ is a dcpo and let $x,y \in D$. The
element $x$ is said to be \emph{way-below} $y$---written as
$x \ll y$---if for every directed subset $A$ of $D$, if
$y \sqsubseteq \join A$, then there exists an element $d \in A$ such
that $x \sqsubseteq d$. An element $x \in D$ is said to be
\emph{finite} if $x \ll x$.

For every element $x$ of a dcpo
$\mathbb{D} \equiv (D, \sqsubseteq)$, let
$\waybelows{x} \defeq \setbarNormal{a \in D}{a \ll x}$. A subset
$B \subseteq D$ is said to be a \emph{basis} for $\mathbb{D}$ if for
every element $x \in D$, the set $B_x \defeq \waybelows{x} \cap B$ is
a directed subset and $x = \join B_x$. A dcpo is said to be
($\omega$-)continuous if it has a (countable) basis, and it is said to
be ($\omega$-)algebraic if it has a (countable) basis consisting
entirely of finite elements.

\begin{definition}[Domain]
  We call $\mathbb{D} \equiv (D, \sqsubseteq)$ a domain if it is a
  pointed continuous dcpo.
\end{definition}

Apart from the order-theoretic structure, domains also have a
topological structure. Assume that
$\mathbb{D} \equiv (D, \sqsubseteq)$ is a poset. A subset
$O \subseteq D$ is said to be \emph{Scott open} if it has the
following properties:
\begin{enumerate}[label=(\arabic*)]
\item It is an upper set, {\ie},
  $\forall x \in O, \forall y \in D: x \sqsubseteq y \implies y \in
  O$.
\item For every directed set $A \subseteq D$ for which $\join A$
  exists, if $\join A \in O$ then $A \cap O \neq \emptyset$.
\end{enumerate}

The collection of all Scott open subsets of a poset $\mathbb{D}$ forms
a $T_0$ topology $\Scott{\mathbb{D}}$, refered to as the Scott
topology. A function $f: \mathbb{D}_1 \to \mathbb{D}_2$ is said to be
Scott continuous if it is continuous with respect to the Scott
topologies on $\mathbb{D}_1$ and $\mathbb{D}_2$. Scott continuity can
be stated purely in order-theoretic terms, {\ie}, a map
$f: (D_1, \sqsubseteq_1) \to (D_2, \sqsubseteq_2)$ between two posets
is Scott continuous if and only if it is monotonic and preserves the
suprema of directed sets, {\ie}, for every directed set
$X \subseteq D_1$ for which $\join X$ exists, we have
$f(\join X) = \join
f(X)$~\cite[Proposition~4.3.5]{Goubault-Larrecq:Non_Hausdorff_topology:2013}.

A poset $(D, \sqsubseteq)$ is said to be a lattice if it is closed
under binary join and binary meet. A lattice $(D, \sqsubseteq)$ is
called:

\begin{itemize}
\item bounded if it has both a bottom element $\bot$ and a top element
  $\top$.

\item complete if $\forall A \subseteq D: \join A$ exists in $D$. Note
  that every complete lattice must be bounded, with $\bot = \join \emptyset$
  and $\top = \join D$.
  
\item distributive if
  $\forall x,y,z \in D: x \binmeet (y \binjoin z) = (x \binmeet y)
  \binjoin (x \binmeet z)$.
\end{itemize}
By a continuous lattice we mean a complete lattice with a basis. Other
variants ({\ie}, $\omega$-continuous, algebraic, and
$\omega$-algebraic) are defined accordingly. Of particular interest to
our discussion are the arithmetic lattices, {\ie}, continuous
distributive lattices $(D, \sqsubseteq)$ with the following property:
\begin{equation*}
  \forall x,y,z \in D: \quad (x \ll y \text{ and } x \ll z) \implies x \ll y
  \binmeet z.
\end{equation*}

For every topological space
$\mathbb{X} \equiv (X, \tau_{\mathbb{X}})$, the poset
$(\tau_{\mathbb{X}}, \subseteq)$ of open subsets of $X$ ordered by
subset relation is a complete distributive lattice, in which
$\bot = \emptyset$ and $\top = X$. Furthermore, we have:
\begin{equation*}
  \forall A \subseteq \tau_{\mathbb{X}}: \quad \join A = \bigcup A \text{ and
  } \meet A = \interiorOf{(\bigcap A)},
\end{equation*}
where $\interiorOf{(\cdot)}$ denotes the interior operator. When the
lattice $(\tau_{\mathbb{X}}, \subseteq)$ is continuous, the
topological space~$\mathbb{X}$ is said to be
\emph{core-compact}. Core-compactness is a desirable property which
guarantees that we obtain `well-behaved' function spaces.

Recall that a dcpo $(D,\sqsubseteq)$ is \emph{bounded-complete} if
each bounded subset $A \subseteq D$ has a join $\join A \in D$. Let
$(D, \sqsubseteq_0)$ be a bounded-complete domain (bc-domain). We let
$\mathbb{D} \equiv (D, \Scott{\mathbb{D}})$ denote the topological
space with the carrier set $D$ endowed with the Scott topology
$\Scott{\mathbb{D}}$. The space $[\mathbb{X} \to \mathbb{D}]$ of
functions $f: X \to D$ which are
$(\tau_{\mathbb{X}}, \Scott{\mathbb{D}})$ continuous can be ordered
pointwise by defining:
\begin{equation*}
  \forall f, g \in [\mathbb{X} \to \mathbb{D}]: \quad f \sqsubseteq g \iff \forall x \in X:
  f(x) \sqsubseteq_0 g(x).
\end{equation*}
It is straightfoward to verify that the poset
$([\mathbb{X} \to \mathbb{D}], \sqsubseteq)$ is directed-complete and
$\forall x \in X: (\join_{i \in I} f_i)(x) = \join
\setbarNormal{f_i(x)}{i \in I}$, for any
$\setbarNormal{f_i}{i \in I} \dirSubsetEq [\mathbb{X} \to
\mathbb{D}]$. By `well-behaved' function spaces we mean those for
which the dcpo $([\mathbb{X} \to \mathbb{D}], \sqsubseteq)$ is
continuous:
\begin{theorem}
  \label{thm:core_compact_bc_domain}
  For any topological space $\mathbb{X}$ and non-singleton bc-domain
  $\mathbb{D}$, the function space
  $([\mathbb{X} \to \mathbb{D}], \sqsubseteq)$ is a bc-domain
  $\iff \mathbb{X}$ is core-compact.
\end{theorem}

\begin{proof}
  For the ($\Leftarrow$) direction,
  see~\cite[Proposition~2]{Erker_et_al:way_below:1998}. A proof
  of the ($\Rightarrow$) direction can also be found
  on~\cite[pages 62 and 63]{Erker_et_al:way_below:1998}.
\end{proof}

The connection between topology and order theory runs much deeper than
stated so far in our discussion. We briefly mention some more results
on this as they will be needed later on, but the interested reader may
refer
to~\cite{AbramskyJung94-DT,Goubault-Larrecq:Non_Hausdorff_topology:2013}
for a more comprehensive account of the connection.

A complete lattice $\mathbb{L}$ is called a frame if it satisfies the
infinite distributivity law
$x \binmeet \join_{i \in I} y_i = \join_{i \in I} (x \binmeet
y_i)$. It is not difficult to verify that, for every topological space
$\mathbb{X} \equiv (X, \tau_{\mathbb{X}})$, the complete lattice
$(\tau_{\mathbb{X}}, \subseteq)$ is indeed a frame. A map
$f: \mathbb{K} \to \mathbb{L}$ between two complete lattices is called
a frame homomorphism if it preserves finite meets and arbitrary
joins. We let $\FrmCat$ denote the category of frames and frame
homomorphisms, and we let $\TopCat$ denote the category of topological
spaces and continuous functions. The map that assigns
$(\tau_{\mathbb{X}}, \subseteq)$ to $\mathbb{X}$ can be extended to a
functor $\Omega: \TopCat \to \opCat{\FrmCat}$ by mapping every
continuous function $f: \mathbb{X} \to \mathbb{Y}$ to
$\Omega(f): \Omega(\mathbb{Y}) \to \Omega(\mathbb{X})$ defined by
$\forall O \in \tau_{\mathbb{Y}}: \Omega(f)(O) = f^{-1}(O)$.

Going in the opposite direction, {\ie}, recovering a topological space
from the lattice $\mathbb{L}$ of its open sets, is the core of Stone
duality. A subset $F$ of a complete lattice
$\mathbb{L} \equiv (L, \sqsubseteq)$ is called a filter if it is
non-empty and satisfies the following two conditions:
\begin{enumerate}[label=(\roman*)]
\item $F$ is an upper set, {\ie}, $\forall x \in F, y \in L: x
  \sqsubseteq y \implies y \in F$.

  \item $F$ is downward directed, {\ie}, $\forall x,y \in F: x
    \binmeet y \in F$.
  \end{enumerate}
  A filter $F \subseteq L$ is said to be completely prime if
  $\forall A \subseteq L: \join A \in F \implies A \cap F \neq
  \emptyset$. Notice the similarity with the definition of Scott open
  sets, except that here we allow $A \subseteq L$ to be arbitrary, not
  just directed. In particular, every completely prime filter is Scott
  open. When $\mathbb{L} = (\tau_{\mathbb{X}}, \subseteq)$, the filter
  of all the open neighborhoods of any given $x \in X$ is completely
  prime. Taking that as a guide, for any complete lattice
  $\mathbb{L} \equiv (L, \sqsubseteq)$, by a point of $\mathbb{L}$ we
  mean a completely prime filter $F \subseteq L$. We let
  $\pt(\mathbb{L})$ denote the set of points of $\mathbb{L}$ with the
  so-called hull-kernel topology, with open sets
  ${\cal O}_u \defeq \setbarNormal{x \in \pt(\mathbb{L})}{ u \in x}$,
  where $u$ ranges over all the elements of
  $L$~\cite[Proposition~8.1.13]{Goubault-Larrecq:Non_Hausdorff_topology:2013}. The
  map $\pt$ can be extended to a functor
  $\pt: \opCat{\FrmCat} \to \TopCat$ as follows: for any morphism
  $g: \mathbb{L} \to \mathbb{K}$ in $\opCat{\FrmCat}$ ({\ie}, a frame
  homomorphism $g: \mathbb{K} \to \mathbb{L}$) the function
  $\pt(g): \pt(\mathbb{L}) \to \pt(\mathbb{K})$ maps every completely
  prime filter $x$ of $\mathbb{L}$ to $g^{-1}(x)$. It is well-known
  that $\pt$ is right adjoint to $\Omega$:
  \begin{equation*}
    \begin{tikzcd}[column sep = large]
      \opCat{\FrmCat} \arrow[r, "\pt", yshift = 1.1ex] & \TopCat  \arrow[l, "\Omega", "\top"', yshift =
      -1.1ex]
    \end{tikzcd}.
  \end{equation*}
  Of particular interest are the cases where this adjunction restricts
  to an equivalence between sub-categories of $\TopCat$ and
  $\opCat{\FrmCat}$. For a detailed account,
  see~\cite[Section~7]{AbramskyJung94-DT}.
  
  Let us call an arithmetic lattice in which the top element is finite
  ({\ie}, $\top \ll \top$) a fully arithmetic lattice. The above
  adjunction restricts to an equivalence 
  \begin{equation}
    \label{eq:diag:Afal_Spec}
    \begin{tikzcd}[column sep = large]
      \opCat{\AFALCat} \arrow[r, "\pt", yshift = 1.1ex] & \SpecCat \arrow[l, "\Omega", "\top"', yshift =
      -1.1ex]
    \end{tikzcd}
  \end{equation}
  between the opposite of the category $\AFALCat$ of algebraic fully
  arithmetic lattices and frame homomorphisms on the one hand, and the
  category $\SpecCat$ of \emph{spectral} spaces and spectral maps on
  the other~\cite[Theorem 7.2.22.]{AbramskyJung94-DT}. A spectral
  space is a compact, sober, coherent, and \emph{strongly locally
    compact} space. This last property is of relevance to our
  discussion later.  Let $\mathbb{Y} = (Y, \tau_{\mathbb{Y}})$ be a
  topological space. We call a subset $Q \subseteq Y$
  \emph{compact-open} if it is both compact and open. A function
  $f: \mathbb{X} \to \mathbb{Y}$ between two spectral spaces is said
  to be a spectral map if for every compact-open $K \subseteq Y$, the
  inverse image $f^{-1}(K)$ is a compact-open subset of $X$.
  
\begin{definition}
  \label{def:strongly_locally_compact}
We say that $\mathbb{Y}$ is a \emph{strongly locally compact}
space if its topology has a base of compact-open subsets.
\end{definition}

\begin{proposition}
  \label{prop:compact_open_finite}
  Let $\mathbb{Y} = (Y, \tau_{\mathbb{Y}})$ be a topological space. A
  set $Q \subseteq Y$ is compact-open if and only if $Q$ is a finite
  element of the complete lattice $(\tau_{\mathbb{Y}}, \subseteq)$,
  {\ie}, $Q \ll Q$.
\end{proposition}

\begin{proof}
Straightforward.  
\end{proof}

\begin{lemma}
  \label{lemma:sLocCompact_AlgCLat}
  A topological space $\mathbb{Y} = (Y, \tau_{\mathbb{Y}})$ is
  strongly locally compact if and only if
  $(\tau_{\mathbb{Y}}, \subseteq)$ is an algebraic lattice.
\end{lemma}

\begin{proof}
  To prove the $(\Rightarrow)$ direction, we first note that for any
  topological space $\mathbb{Y} = (Y, \tau_{\mathbb{Y}})$, the lattice
  $(\tau_{\mathbb{Y}}, \subseteq)$ is complete. So, it remains to show
  that $(\tau_{\mathbb{Y}}, \subseteq)$ is algebraic. Take any open
  set $O \in \tau_{\mathbb{Y}}$. As $\mathbb{Y}$ is assumed to be
  strongly locally compact, for each $x \in O$, there exists a
  compact-open $Q_x \in \tau_{\mathbb{Y}}$ such that
  $x \in Q_x \subseteq O$, which implies that
  $O = \bigcup_{x \in O} Q_x$. Hence, by
  Proposition~\ref{prop:compact_open_finite}, the set of all
  compact-open subsets of $Y$ forms a basis of finite elements for
  $(\tau_{\mathbb{Y}}, \subseteq)$.

  To prove the $(\Leftarrow)$ direction, take any open set
  $O \in \tau_{\mathbb{Y}}$. As $(\tau_{\mathbb{Y}}, \subseteq)$ is
  assumed to be algebraic, then
  $O = \join \setbarNormal{Q \ll O}{Q \text{ is finite}}$, which
  implies that
  $\forall x \in O: \exists Q_x \ll O: x \in Q_x \subseteq O$ and
  $Q_x$ is a finite element of $(\tau_{\mathbb{Y}}, \subseteq)$. By
  Proposition~\ref{prop:compact_open_finite}, this $Q_x$ must be
  compact-open.
\end{proof}

In this article, we start our construction with bounded distributive
lattices of open sets, which are closely related to algebraic fully
arithmetic lattices. Assuming that $\mathbb{L}_1$ and $\mathbb{L}_2$
are two bounded distributive lattices, a map
$f: \mathbb{L}_1 \to \mathbb{L}_2$ is said to be a \emph{bounded
  lattice homomorphism} if it is a monotone map that preserves all
finite joins and all finite meets. Let $\BDLatCat$ denote the category
of bounded distributive lattices and bounded lattice
homomorphisms. Then, there is an equivalence of categories:
\begin{equation}
  \label{eq:diag:Afal_BDLat}
    \begin{tikzcd}[column sep = large]
      \AFALCat \arrow[r, "{\cal K}", yshift = 1.1ex] & \BDLatCat
      \arrow[l, "\Idl", "\top"', yshift = -1.1ex]
    \end{tikzcd},
  \end{equation}
  in which, for any algebraic fully arithmetic lattice $\mathbb{L}$,
  ${\cal K}(\mathbb{L})$ is the bounded distributive lattice of finite
  elements of $\mathbb{L}$. The functor
  $\Idl : \BDLatCat \to \AFALCat$ is ideal completion.
  
  By composing diagrams~\eqref{eq:diag:Afal_Spec}
  and~\eqref{eq:diag:Afal_BDLat}, we obtain:
    \begin{equation}
    \label{eq:diag:BDLat_Spec}
    \begin{tikzcd}[column sep = large]
      \opCat{\BDLatCat} \arrow[r, "\opCat{\Idl}", yshift = 1.1ex] & \opCat{\AFALCat} \arrow[l, "\opCat{{\cal K}}", "\top"', yshift =
      -1.1ex] \arrow[r, "\pt", yshift = 1.1ex] & \SpecCat \arrow[l, "\Omega", "\top"', yshift =
      -1.1ex]
    \end{tikzcd}.
  \end{equation}
  \noindent
  A detailed account of this equivalence, in the framework of the
  current article, will be presented in
  Section~\ref{sec:Core_Compactification}.

  %%%%%%%%%%%%%%%%%%%%%%%%%%%%%%%%%%%%%%%%%%%%%%%%%%%%%%%%%
  \section{Basic Galois Connection}
\label{sec:Basic_Galois_Connection}

Assume that $\mathbb{X} \equiv (X, \tau_{\mathbb{X}})$ and
$\mathbb{Y} \equiv (Y, \tau_{\mathbb{Y}})$ are two topological spaces.

\begin{definition}[Quasi-embedding, Embedding]
  \label{def:quasi_embedding}
  A continuous map $\iota: \mathbb{X} \to \mathbb{Y}$ is said to be:
\begin{enumerate}[label=(\roman*)]
\item a quasi-embedding of $\mathbb{X}$ into $\mathbb{Y}$ if it is
  \emph{relatively open}, {\ie},
  $\forall U \in \tau_{\mathbb{X}}: \exists V \in \tau_{\mathbb{Y}}: U
  = \iota^{-1}(V)$.

\item an embedding of $\mathbb{X}$ into $\mathbb{Y}$ if it is an
  injective quasi-embedding.
\end{enumerate}
\end{definition}
\noindent
Also, recall that the map $\iota$ is said to be dense if $\iota(X)$ is
a dense subset of $Y$.

Let $\TopCat_0$ denote the category of $T_0$ topological spaces and
continuous maps.  Over $T_0$ spaces, the two notions of
Definition~\ref{def:quasi_embedding} coincide:

\begin{proposition}
  Assume that $\mathbb{X}:\TopCat_0$, $\mathbb{Y}: \TopCat$, and
  $\iota: \mathbb{X} \to \mathbb{Y}$ is relatively open. Then $\iota$
  must be injective. As a consequence, when $\mathbb{X}: \TopCat_0$,
  every quasi-embedding $\iota: \mathbb{X} \to \mathbb{Y}$ is an
  embedding.
\end{proposition}

\begin{proof}
  Let $x_1 \neq x_2 \in X$. Since $\mathbb{X}$ is assumed to be $T_0$,
  without loss of generality, we assume that there exists an open
  $U \in \tau_{\mathbb{X}}$ such that $x_1 \in U$ and
  $x_2 \not \in U$. As $\iota$ is relatively open, there exists an
  open $V \in \tau_{\mathbb{Y}}$ such that $U = \iota^{-1}(V)$. Hence,
  $\iota(x_1) \in V$ while $\iota(x_2) \not \in V$. Therefore,
  $\iota(x_1) \neq \iota(x_2)$.
\end{proof}

Let us assume that $\mathbb{D}$ is a bc-domain with Scott
topology. If $\iota: \mathbb{X} \to \mathbb{Y}$ is a dense
quasi-embedding, then, as we will see, the two function spaces
$[\mathbb{X} \to \mathbb{D}]$ and $[\mathbb{Y} \to \mathbb{D}]$ are
related via a \emph{Galois connection}:

\begin{definition}[Category $\PoCat$, Galois connection $F \dashv G$]
  We let $\PoCat$ denote the category of posets and monotonic
  maps. A Galois connection in the category $\PoCat$ between two
  posets $\mathbb{C} \equiv (C, \sqsubseteq_{\mathbb{C}})$ and
  $\mathbb{D} \equiv (D, \sqsubseteq_{\mathbb{D}})$ is a pair
  of monotonic maps:
  \begin{equation*}
    \begin{tikzcd}[column sep = large]
      \mathbb{D} \arrow[r, "G", yshift = 1.1ex] & \mathbb{C} \arrow[l,
      "F", "\top"', yshift = -1.1ex]
    \end{tikzcd}
  \end{equation*}
  such that:
  \begin{equation*}
  \forall x \in C: \forall y \in D: F(x)
  \sqsubseteq_{\mathbb{D}} y \iff x \sqsubseteq_{\mathbb{C}} G(y).  
\end{equation*}
In this case, we call $F: \mathbb{C} \to \mathbb{D}$ the left adjoint
and $G: \mathbb{D} \to \mathbb{C}$ the right adjoint, and write
$F \dashv G$.
\end{definition}

Our aim is to show that the two function spaces are related via the
following Galois connection:
  \begin{equation*}
    \begin{tikzcd}[column sep = large]
      [\mathbb{X} \to \mathbb{D}] \arrow[r, "(\cdot)_*", yshift =
      1.1ex] & {[ \mathbb{Y} \to \mathbb{D}]}  \arrow[l, "(\cdot)^*", "\top"', yshift = -1.1ex],
    \end{tikzcd}
  \end{equation*}
  in which:
  \begin{equation}
    \label{eq:g_*_embedding}
    \forall g \in [\mathbb{Y} \to \mathbb{D}]:  \quad g^* \defeq g
    \circ \iota ,
  \end{equation}
  and:
  \begin{equation}
    \label{eq:f_*_embedding}
    \forall f \in [\mathbb{X} \to \mathbb{D}]: \forall  y \in
    Y:\quad  f_*(y) \defeq \join
    \setbarTall{\meet f(\iota^{-1}(U))}{ y \in U \in \tau_{\mathbb{Y}}}.
  \end{equation}

\begin{equation*}
  \begin{tikzcd}[row sep = large, column sep = large]
\mathbb{X} \arrow[r, "\iota"] \arrow[dr, "g*"', color = blue] & \mathbb{Y} \arrow[d,
    "g"] &\\
     & \mathbb{D} 
  \end{tikzcd}
  \begin{tikzcd}[row sep = large, column sep = large]
\mathbb{X} \arrow[r, "\iota"] \arrow[dr, "f"'] & \mathbb{Y} \arrow[d,
    "f_*", color = blue] &\\
    & \mathbb{D}
  \end{tikzcd}    
\end{equation*}
\noindent
The formulation of the map $(\cdot)_*$ in~\eqref{eq:f_*_embedding} is
a variation of a common construction as appears in, {\eg},
\cite[Exercise~II-3.19]{Gierz-ContinuousLattices-2003}. The map $f_*$
is sometimes referred to as the envelope of $f$ (see, {\eg},
\cite{DiGianantonio_Edalat_Gutin-Automatic_Differentiation:2023}).

\begin{proposition}
  The map $(\cdot)^*: [\mathbb{Y} \to \mathbb{D}] \to
  [\mathbb{X} \to \mathbb{D}]$ is well-defined, {\ie}:
  \begin{equation*}
    \forall g \in [\mathbb{Y} \to \mathbb{D}]:  g^* \in [\mathbb{X} \to \mathbb{D}].    
  \end{equation*}
  \end{proposition}
  
  \begin{proof}
    Follows from the fact that $\iota$ is continuous.
  \end{proof}
  
\begin{proposition}
  For all $f \in [\mathbb{X} \to \mathbb{D}]$ and $y \in Y$, $f_*(y)$
  is well-defined.
\end{proposition}

\begin{proof}
  As the quasi-embedding $\iota$ is dense, for all non-empty
  $U \in \tau_{\mathbb{Y}}: \iota^{-1}(U) \neq \emptyset$, which
  implies that $f(\iota^{-1}(U)) \neq \emptyset$. Since $D$ is assumed
  to be a bc-domain, it must have the infima of all non-empty
  subsets~\cite[Proposition~4.1.2]{AbramskyJung94-DT}. Hence,
  $\meet f(\iota^{-1}(U))$ exists.
  
  Next, we must show that the set
  $A \defeq \setbarTall{\meet f(\iota^{-1}(U))}{ y \in
    U \in \tau_{\mathbb{Y}}}$ is directed. Take
  $U_1, U_2 \in \tau_{\mathbb{Y}}$ such that
  $y \in U_1$ and $y \in U_2$. Then,
  $y \in U_1 \cap U_2$, hence
  $U_1 \cap U_2 \neq \emptyset$, and
  $\meet f(\iota^{-1}(U_1 \cap U_2)) \in A$ is an upper bound of
  both $\meet f(\iota^{-1}(U_1))$ and $\meet
  f(\iota^{-1}(U_2))$. Therefore, $A$ is directed and $\join A$
  exists.
\end{proof}

\begin{proposition}
  The map
  $(\cdot)_*: [\mathbb{X} \to \mathbb{D}] \to [\mathbb{Y} \to
  \mathbb{D}]$ is well-defined, {\ie}, $    \forall f \in [\mathbb{X} \to \mathbb{D}]: f_* \in [\mathbb{Y} \to \mathbb{D}]$.
\end{proposition}
  
  \begin{proof}
    We must prove that for all $f \in [\mathbb{X} \to \mathbb{D}]$,
    $f_*$ is continuous with respect to $\tau_{\mathbb{Y}}$ and the
    Scott topology~$\Scott{\mathbb{D}}$. It suffices to show that, for
    any $e \in D$, the set $f_*^{-1}(\wayaboves e)$ is open. Take any
    $y \in Y$. Then:
    \begin{align*}
      y \in f_*^{-1}(\wayaboves e) \implies & e \ll f_*(y)
      \\
      (\text{by~\eqref{eq:f_*_embedding}}) \implies & \exists
                                                               U_0 \in
                                                               \tau_{\mathbb{Y}}:
                                                               e \ll
                                                               \meet
                                                               f(\iota^{-1}(U_0))
                                                      \text{ and } y
                                                      \in U_0.
    \end{align*}
    Note that $U_0$ is an open neighborhood of
    $y$. We claim that $U_0 \subseteq
    f_*^{-1}(\wayaboves e)$. Take any arbitrary $\hat{y} \in U_0$. Then:
    \begin{equation*}
      e \ll \meet f( \iota^{-1}(U_0)) \sqsubseteq \join
    \setbarTall{\meet f(\iota^{-1}(U))}{ \hat{y} \in U \in \tau_{\mathbb{Y}}} = f_*(\hat{y}).
  \end{equation*}
  Hence, $\hat{y} \in f_*^{-1}(\wayaboves e)$.
\end{proof}

\begin{proposition}
  \label{prop:f_and_f_*}
  For every $f \in [\mathbb{X} \to \mathbb{D}]$, we have $f = f_* \circ
  \iota$. In particular, $f = (f_*)^*$.

\begin{equation*}
  \begin{tikzcd}[row sep = large, column sep = large]
\mathbb{X} \arrow[r, "\iota"] \arrow[dr, "f=(f_*)^*"'] & \mathbb{Y} \arrow[d,
    "f_*"] &\\
     & \mathbb{D} 
  \end{tikzcd}
\end{equation*}  
\end{proposition}

\begin{proof}
  Assume that $x \in X$. For any $U \in \tau_{\mathbb{Y}}$ satisfying
  $\iota(x) \in U$ we have $x \in \iota^{-1}(U)$. As a result,
  $\meet f(\iota^{-1}(U)) \sqsubseteq f(x)$, which implies that
  $f_*(\iota(x)) = \join \setbarTall{\meet f(\iota^{-1}(U))}{ \iota(x)
    \in U \in \tau_{\mathbb{Y}}} \sqsubseteq f(x)$.
  
    Next, we show that $f(x) \sqsubseteq f_*(\iota(x))$. Take an
    arbitrary $e \ll f(x)$. Hence, $f^{-1}(\wayaboves e)$ is an open
    neighborhood of $x$. As $\iota: \mathbb{X} \to \mathbb{Y}$ is
    relatively open, for some $U \in \tau_{\mathbb{Y}}$ we have
    $f^{-1}(\wayaboves e) = \iota^{-1}(U)$. Also, note that
    $\iota(x) \in U$. Thus, we have:
  \begin{align*}
    \left( \forall y \in \iota^{-1}(U): e \ll f( y)\right)
    \implies & e \sqsubseteq \meet f(\iota^{-1}(U))\\
    (\text{by~\eqref{eq:f_*_embedding} and $\iota(x) \in U$})
    \implies & e \sqsubseteq f_*(\iota(x)).
  \end{align*}
  In conclusion, $f(x) \sqsubseteq f_*(\iota(x))$.
\end{proof}

Let us briefly recall the concept of a monotone section-retraction
pair.  Assume that $\mathbb{D}$ and $\mathbb{E}$ are two posets. A
pair of maps $s: \mathbb{D} \to \mathbb{E}$ and
$r: \mathbb{E} \to \mathbb{D}$ is called a monotone section-retraction
pair if $s$ and $r$ are monotone and $r \circ s =
\id_{\mathbb{D}}$. In this case, $\mathbb{D}$ is said to be a monotone
retract of $\mathbb{E}$. It is straightforward to verify that if $s$
and $r$ form a section-retraction pair, then $s$ must be injective and
$r$ must be surjective. For a more detailed account of
section-retraction pairs the reader may refer
to~\cite[Section~3.1.1]{AbramskyJung94-DT}.

\begin{lemma}
  \label{lemma:section_retraction}
  The maps $(\cdot)_*$ and $(\cdot)^*$ form a monotone section
  retraction pair between $[\mathbb{X} \to \mathbb{D}]$ and $[\mathbb{Y} \to \mathbb{D}]$.
\end{lemma}

\begin{proof}
  Monotonicity of $(\cdot)^*$ follows from that of composition, {\ie}:
  \begin{equation*}
   \forall g, g' \in [\mathbb{Y} \to \mathbb{D}]: \quad g \sqsubseteq g'
  \implies g \circ \iota \sqsubseteq g' \circ \iota. 
  \end{equation*}
  Monotonicity of $(\cdot)_*$ follows from monotonicity of meet and
  join. To be more precise, take
  $f, f' \in [\mathbb{X} \to \mathbb{D}]$ satisfying
  $f \sqsubseteq f'$, and assume that $y \in Y$. Then,
  $\forall U \in \tau_{\mathbb{Y}}: \meet f(\iota^{-1}(U)) \sqsubseteq
  \meet f'(\iota^{-1}(U))$. Thus,
  $\join \setbarTall{\meet f(\iota^{-1}(U))}{ y \in U \in
    \tau_{\mathbb{Y}}} \sqsubseteq \join \setbarTall{\meet
    f'(\iota^{-1}(U))}{y \in U \in \tau_{\mathbb{Y}}}$, which,
  by~\eqref{eq:f_*_embedding}, implies that
  $f_*(y) \sqsubseteq f'_*(y)$.

By Proposition~\ref{prop:f_and_f_*},
$\forall f \in [\mathbb{X} \to \mathbb{D}]: f = (f_*)^*$. Therefore,
$(\cdot)_*$ and $(\cdot)^*$ form a section retraction pair.
\end{proof}

\begin{theorem}[Galois connection]
  \label{thm:Galois_connection_X_to_D}
  Assume that $\mathbb{X}$ and $\mathbb{Y}$ are two topological
  spaces, $\iota: \mathbb{X} \to \mathbb{Y}$ is a dense
  quasi-embedding, and $\mathbb{D}$ is a bc-domain. Then, the maps
  $(\cdot)^*$ and $(\cdot)_*$ (defined in~\eqref{eq:g_*_embedding} and
  \eqref{eq:f_*_embedding}, respectively) form a Galois connection:
  \begin{equation}
    \label{eq:Galois_Connection_QEmbed}
    \begin{tikzcd}[column sep = large]
      [\mathbb{X} \to \mathbb{D}] \arrow[r, "(\cdot)_*", yshift = 1.1ex] & {[\mathbb{Y} \to \mathbb{D}]}  \arrow[l, "(\cdot)^*", "\top"', yshift = -1.1ex]
    \end{tikzcd}
  \end{equation}
  in the category $\PoCat$, in which, $(\cdot)_*$ is the right adjoint,
  and $(\cdot)^*$ is the left adjoint. Furthermore:
  \begin{enumerate}[label=(\roman*)]

  \item \label{item:epi_mono_X_to_D} The map $(\cdot)^*$ is
    surjective, and $(\cdot)_*$ is injective.
    
  \item \label{item:Galois_id_X_to_D} $(\cdot)^* \circ (\cdot)_* = \id_{[\mathbb{X}
      \to \mathbb{D}]}$, {\ie},
    $\forall f \in [\mathbb{X} \to \mathbb{D}]: (f_*)^* = f$.

  \item \label{item:left_adjoint_Scott_cont_X_to_D} The left adjoint
    $(\cdot)^*$ is Scott continuous.
  \end{enumerate}
\end{theorem}

\begin{proof}
  Claims~\ref{item:epi_mono_X_to_D} and~\ref{item:Galois_id_X_to_D}
  follow from Lemma~\ref{lemma:section_retraction}.  For any
  adjunction between two dcpos, the left adjoint is Scott
  continuous~\cite[Proposition~3.1.14]{AbramskyJung94-DT}. So, given
  that both $[\mathbb{X} \to \mathbb{D}]$ and
  $[\mathbb{Y} \to \mathbb{D}]$ are dcpos,
  claim~\ref{item:left_adjoint_Scott_cont_X_to_D} will also be
  established once we prove that the maps $(\cdot)^*$ and $(\cdot)_*$
  form a Galois connection.
  
  To that end, we must prove that, for any $f \in [\mathbb{X} \to \mathbb{D}]$
  and
  $g \in [\mathbb{Y} \to \mathbb{D}]: g^* \sqsubseteq f \iff
  g \sqsubseteq f_*$. Equivalently:
  \begin{equation*}
    \forall x \in X: g( \iota(x)) \sqsubseteq f(x) \iff \forall
    y \in Y: g( y) \sqsubseteq f_*(y).
  \end{equation*}
  To prove the $(\Leftarrow)$ implication, for any given $x \in X$, by
  assumption, $g( \iota(x)) \sqsubseteq
  f_*(\iota(x))$. This, together with Proposition~\ref{prop:f_and_f_*}
  imply $\forall x \in X: g( \iota(x)) \sqsubseteq f(x)$.
  
  To prove the $(\Rightarrow)$ implication, assume that
  $y \in Y$ is given, and consider an arbitrary
  $e \ll g(y)$. From the assumption, we obtain
  $\forall x \in \iota^{-1}( g^{-1}(\wayaboves e)): g( \iota( x))
  \sqsubseteq f(x)$, which implies that:
  \begin{equation}
    \label{eq:meet_g_iota_f}
    \meet_{x \in \iota^{-1}( g^{-1} (\wayaboves e))} g( \iota(x))
    \sqsubseteq \meet f\left(\iota^{-1} (g^{-1}( \wayaboves e))\right).
  \end{equation}
  On the other hand, for any
  $x \in \iota^{-1} (g^{-1}( \wayaboves e))$, we have
  $e \ll g(\iota(x))$, which, together with~\eqref{eq:meet_g_iota_f},
  implies
  $e \sqsubseteq \meet f\left(\iota^{-1} (g^{-1}( \wayaboves
    e))\right)$. As $y \in g^{-1}(\wayaboves e)$, we obtain
  $e \sqsubseteq \join \setbarTall{\meet f(\iota^{-1}(U))}{
    y \in U \in \tau_{\mathbb{Y}}} = f_*(y)$, where, in
  the last equality, we have used~\eqref{eq:f_*_embedding}. As
  $e \ll g(y)$ was arbitrary, we conclude that $g(y)
  \sqsubseteq f_*(y)$.  
\end{proof}

%%%%%%%%%%%%%%%%%%%%%%%%%%%%%%%%%%%%%%%%%%%%%%%%%%%%%%%%%
\section{Core-Compactification via Spectral Spaces}
\label{sec:Core_Compactification}

In applications, the Galois connection of
Theorem~\ref{thm:Galois_connection_X_to_D} is useful when one of the
spaces has certain desirable properties that the other does not have,
for instance, when $\mathbb{X}$ is not core-compact, but $\mathbb{Y}$
is. Recall from Theorem~\ref{thm:core_compact_bc_domain} that,
whenever $\mathbb{D}$ is a bc-domain and $\mathbb{Y}$ is
core-compact, then $[\mathbb{Y} \to \mathbb{D}]$ is also a
bc-domain. In fact, there is an explicit description of a basis
for $[\mathbb{Y} \to \mathbb{D}]$ consisting of \emph{step functions},
as we will explain below.

Assume that $\mathbb{X} \equiv (X, \tau_{\mathbb{X}})$ is a
topological space, and
$\mathbb{D} \equiv (D, \sqsubseteq )$ is a pointed
directed-complete partial order (pointed dcpo), with bottom
element~$\bot$. Then, for every open set $O \in \tau_{\mathbb{X}}$,
and every element $b \in D$, we define the single-step function
$b \chi_O: X \to D$ as follows:
\begin{equation*}
%  \label{eq:single_step_fun}
  b \chi_O (x) \defeq
  \left\{
    \begin{array}{ll}
      b, &  \quad \text{if } x \in O,\\
      \bot, & \quad \text{if } x \in X \setminus O.
    \end{array}
  \right.
\end{equation*}  
For any set $\setbarTall{b_i \chi_{O_i}}{i \in I}$ of
single-step functions, the supremum $\join_{i \in I} b_i \chi_{O_i}$
exists if and only if $\setbarTall{b_i \chi_{O_i}}{i \in I}$ satisfies
the following consistency condition:
\begin{equation*}
  \forall J \subseteq I: \quad \bigcap_{j \in J} O_j \neq \emptyset
  \implies \exists b_J \in D: \forall j \in J: b_j \sqsubseteq b_J.
\end{equation*}
By a step-function we mean the join of a consistent finite set of
single-step functions.

\begin{lemma}
  \label{lemma:generalized_erker_sup_step_fun}
  Assume that:
  \begin{enumerate}[label=(\arabic*)]

  \item $\mathbb{Y} = (Y, \tau_{\mathbb{Y}})$ is a core-compact space
    and $B_{\mathbb{Y}}$ is a basis for the continuous lattice
    $(\tau_{\mathbb{Y}}, \subseteq)$.

  \item $\mathbb{D}$ is a bc-domain and $D_0 \subseteq D$ is a
    basis for $\mathbb{D}$.
  \end{enumerate}
  Then:
  \begin{equation}
    \label{eq:generalized_erker_sup_step_fun}
  \forall f \in [\mathbb{Y} \to \mathbb{D}]:\quad  f
  = \join \setbarTall{b\chi_{U}}{U \ll f^{-1}(\wayaboves{b}), U \in
    B_{\mathbb{Y}}, b \in D_0 }.
\end{equation}
In particular, $[\mathbb{Y} \to \mathbb{D}]$ is a
    bc-domain with a basis $\mathbb{B}$ of step-functions of the
    form:
  \begin{equation*}
    %\label{eq:basis_B_Core_Compact}
    \mathbb{B} = \setbarTall{\join_{i \in I} b_i \chi_{U_i}}{
      \text{$I$ is finite, $\setbarTall{b_i \chi_{U_i}}{i \in I}$ is
        consistent}, \forall i \in I: U_i \in B_{\mathbb{Y}}, b_i \in D_0}.
  \end{equation*}
  \end{lemma}

  \begin{proof}
    The proof is a straightforward modification of the proof
    of~\cite[Lemma~1(c)]{Erker_et_al:way_below:1998}.     Note that,
    as $\mathbb{Y}$ is assumed to be core-compact, the lattice
    $(\tau_{\mathbb{Y}}, \subseteq)$ is continuous. Furthermore, any
    basis ({\eg}, $B_{\mathbb{Y}}$) of the continuous lattice
    $(\tau_{\mathbb{Y}}, \subseteq)$ is a topological base of
    $\mathbb{Y}$. With that in mind, we have:
    \begin{equation}
      \label{eq:b_ll_f_x_implications}
      \forall y \in Y, \forall b \in D: \quad b \ll f(y) \iff y \in
      f^{-1}(\wayaboves{b}) \iff \exists U \in B_{\mathbb{Y}}: y \in U \ll f^{-1}(\wayaboves{b}),
    \end{equation}
    Thus:
    \begin{align*}
\forall y \in Y:\quad      \join \setbarTall{b\chi_{U}}{U \ll f^{-1}(\wayaboves{b}), U \in
    B_{\mathbb{Y}}, b \in D_0 }(y) &= \join \setbarTall{b \in
                                     D_0}{\exists U \in
                                     B_{\mathbb{Y}}: y \in U  \ll
                                     f^{-1}(\wayaboves{b}) } \\
(\text{by~\eqref{eq:b_ll_f_x_implications}})      &= \join
                                                    \setbarTall{b \in
                                                    D_0}{b \ll f(y)}\\
      (\text{since $D_0$ is a basis of $D$}) &= f(y).
    \end{align*}
  \end{proof}

  Going back to the Galois connection of
  Theorem~\ref{thm:Galois_connection_X_to_D}, we know that the left
  adjoint is always Scott-continuous. Nonetheless, when $\mathbb{Y}$
  is core-compact and $\mathbb{X}$ is not, the right adjoint cannot be
  Scott-continuous:
  
  \begin{proposition}
    \label{prop:right_adjoint_Scott_X_core_compact}
    Assume that $\mathbb{Y}$ is core-compact and $\mathbb{D}$ is a
    bc-domain. If the right adjoint $(\cdot)_*$ in the Galois
    connection~\eqref{eq:Galois_Connection_QEmbed} is Scott continuous
    and $D$ is not a singleton, then $\mathbb{X}$ must be
    core-compact.
\end{proposition}

\begin{proof}
  By Lemma~\ref{lemma:section_retraction} and
  Theorem~\ref{thm:Galois_connection_X_to_D}, the pair
  $\left( (\cdot)_*, (\cdot)^* \right)$ forms a monotone
  section-retraction, with a Scott continuous retraction map
  $(\cdot)^*$. When the section $(\cdot)_*$ is also Scott continuous,
  the dcpo $[\mathbb{X} \to \mathbb{D}]$ becomes a continuous
  retract of the continuous domain $[\mathbb{Y} \to
  \mathbb{D}]$. By~\cite[Theorem~3.1.4]{AbramskyJung94-DT}, any continuous
  retract of a continuous domain is also a continuous domain, hence
  $[\mathbb{X} \to \mathbb{D}]$ must be a continuous domain. By
  Theorem~\ref{thm:core_compact_bc_domain}, this means that
  $\mathbb{X}$ must be core-compact.
\end{proof}

We now establish a working definition of what constitutes a
core-compactification:

\begin{definition}[Core-compactification]
  \label{def:core_compactification}
  We say that a core-compact space $\mathbb{X'}$ is a
  \emph{core-compactification} of the topological space $\mathbb{X}$
  if $\mathbb{X}$ can be embedded as a dense sub-space of
  $\mathbb{X}'$.
\end{definition}

Some classical compactification methods yield core-compactifications
of topological spaces. For instance:

\begin{enumerate}[label=(\roman*)]
\item Let $\mathbb{X}$ be the set $\R^n$ with the Euclidean
  topology. Then, the one-point (Alexandroff) compactification $\R^n
  \cup \set{\infty}$ is a core-compactification of $\R^n$.
\item Let $\mathbb{X}$ be any Tychonoff space. Then, the Stone-{\v
    C}ech compactification $\beta \mathbb{X}$ is a core-compactification of
  $\mathbb{X}$~\cite[Chapter 5]{Munkres:Topology:2000}.
\end{enumerate}

The classical compactification methods, however, are not suitable for
our objectives. For instance, the one-point compactification provides
the right result only when applied to locally compact (hence,
core-compact) spaces, whereas our focus here is mainly on
non-core-compact spaces, even though the method that we present is
applicable to all $T_0$ spaces. The Stone-{\v C}ech compactification
leads to a dense embedding when applied to any Tychonoff space, but an
explicit description of the resulting space $\beta \mathbb{X}$ is
lacking even for simple topological spaces $\mathbb{X}$, and as such,
it is not suitable for computational purposes. The compactification
obtained by our method, on the other hand, is designed for
computational purposes.

Assume that $\mathbb{X} \equiv (X, \tau_{\mathbb{X}})$ is a
topological space. We say that $\Omega_0 \subseteq \tau_{\mathbb{X}}$
is a \emph{ring of open subsets of $\mathbb{X}$} if it is closed under
finite unions and finite intersections. In particular, $\Omega_0$ must
contain $\emptyset = \cup \emptyset$ and $X = \cap \emptyset$. We say
that $\Omega_0$ \emph{separates points} if
$\forall x, y \in X: \exists O \in \Omega_0: (x \in O \wedge y \notin
O) \vee (x \notin O \wedge y \in O)$.

\begin{definition}[Viable base]
  Assume that $\mathbb{X} \equiv (X, \tau_{\mathbb{X}})$ is a
  topological space. We say that
  $\Omega_0 \subseteq \tau_{\mathbb{X}}$ is a \emph{viable base} for
  $\mathbb{X}$ if it is a ring of open sets that forms a base for the
  topology $\tau_{\mathbb{X}}$.
\end{definition}

\begin{remark}
  For every topological space
  $\mathbb{X} \equiv (X, \tau_{\mathbb{X}})$, there is always at least
  one viable base, {\ie}, $\Omega_0 = \tau_{\mathbb{X}}$.
\end{remark}

\begin{proposition}
  If $\mathbb{X}$ is a $T_0$ topological space and $\Omega_0$ is a
  viable base, then $\Omega_0$ separates points.
\end{proposition}

\begin{proof}
  Assume that $x \neq y \in X$. As $\tau_{\mathbb{X}}$ is $T_0$, there
  exists an open set $U \in \tau_{\mathbb{X}}$ that separates $x$ and
  $y$. Without loss of generality, let us assume that $x \in U$ and $y
  \notin U$. Since $\Omega_0$ is a base of the topology, there exists
  a $U' \in \Omega_0$ such that $x \in U' \subseteq U$. Hence, $y
  \notin U'$.
\end{proof}

Any ring $\Omega_0$ of open sets is a bounded distributive lattice
with $\meet \defeq \bigcap$ and $\join \defeq \bigcup$. Assume that
$\Omega_0$ is a ring of open sets and let
${\cal L} \defeq \Idl(\Omega_0)$ be the ideal completion of
$(\Omega_0, \subseteq)$. From diagram~\eqref{eq:diag:BDLat_Spec}, we
know that ${\cal L}$ must be an algebraic fully arithmetic lattice.
The principal ideals $\lowerSet{\emptyset} = \set{\emptyset}$ and
$\lowerSet{X} = \Omega_0$ are the bottom and the top elements of
${\cal L}$, respectively. We mention some basic properties of
${\cal L}$.

\begin{proposition}
  \label{prop:ideals_closed_finite_union}
  Every ideal $I \in {\cal L}$ is closed under finite unions.
\end{proposition}

\begin{proof}
  Assume that $x, y \in I$. As $I$ is an ideal, then
  $\exists z \in I: x \subseteq z$ and $y \subseteq z$. Hence,
  $x \cup y \subseteq z$. On the other hand, as $\Omega_0$ is closed
  under finite unions and $I$ is a lower set, we must have
  $x \cup y \in I$.
\end{proof}

\begin{proposition}
  \label{prop:intersection_ideals_open_sets}
  Assume that $I_1, I_2 \in {\cal L}$. Then, $\forall O_1 \in I_1, O_2
  \in I_2: O_1 \cap O_2 \in I_1 \cap I_2$.
\end{proposition}

\begin{proof}
  Follows from the fact that $I_1$ and $I_2$ are lower sets and
  $\Omega_0$ is closed under finite intersections.
\end{proof}

\begin{proposition}
  \label{prop:L_complet_lattice}
  ${\cal L}$ is a complete lattice, in which, for every subset
  $A \subseteq {\cal L}$, we have $\meet A = \bigcap A$.
\end{proposition}

\begin{proof}
  If $A = \emptyset$, then both sides are equal to the top element of
  the lattice, {\ie},
  $\meet A = \bigcap A = \lowerSet{X} = \Omega_0$. Next, assume
  that $A = \setbarNormal{I_j}{j \in J}$ for some non-empty index set
  $J$. We must show that $\hat{I} \defeq \bigcap_{j \in J} I_j$ is
  indeed an ideal:

  \begin{itemize}
  \item $\hat{I}$ is non-empty: This is because
    $\forall j \in J: \emptyset \in I_j$.
  \item $\hat{I}$ is a lower set: This follows immediately from the
    fact that $\forall j \in J: I_j$ is an ideal.

  \item $\hat{I}$ is directed: Let $x, y \in \bigcap_{j \in J}
    I_j$. By Proposition~\ref{prop:ideals_closed_finite_union}, $\forall j \in J: x \cup y \in I_j$, which
    implies that $x \cup y \in \bigcap_{j \in J} I_j$.
  \end{itemize}

  The ideal $\hat{I}$ is a lower bound of
  $\setbarNormal{I_j}{j \in J}$ because
  $\forall j \in J: \hat{I} \subseteq I_j$. Furthermore, if
  $\tilde{I} \in {\cal L}$ is any other lower bound of
  $\setbarNormal{I_j}{j \in J}$, then
  $\forall j \in J: \tilde{I} \subseteq I_j$, which implies that
  $\tilde{I} \subseteq \bigcap_{j \in J} I_j = \hat{I}$. Therefore, $\hat{I}$
  is the infimum of $\setbarNormal{I_j}{j \in J}$.    
\end{proof}

In a complete lattice, suprema can be obtained using infima. To be
more precise, for any subset $A \subseteq {\cal L}$, let
$\check{A} \defeq \setbarNormal{z \in {\cal L}}{\forall x \in A: x
  \subseteq z}$ be the set of all upper bounds of $A$. Then, we have
$\join A = \meet \check{A}$. In the following proposition, however, we
present an explicit description of the suprema of subsets of
${\cal L}$:\footnote{Also, see~\cite[Exercise~9.5.11]{Goubault-Larrecq:Non_Hausdorff_topology:2013}.}

\begin{proposition}
  \label{prop:join_formula}
  In the complete lattice ${\cal L}$, for every non-empty subset
  $A = \setbarNormal{I_j}{j \in J}$, we have:
  \begin{equation}
    \label{eq:lub_explicit}
    \join A = \bigcup \setbarTall{\lowerSet{\cup_{j \in J_0}
        O_j}}{J_0 \subseteq_f J, \forall j \in J_0: O_j \in I_j}. 
  \end{equation}
  In other words, for any finite combination of open sets
  $\setbarNormal{O_j}{j \in J_0}$ taken from the ideals in $A$, we
  form the principal ideal $\lowerSet{\cup_{j \in J_0}
        O_j}$, and finally take the union of all these principal ideals.
\end{proposition}

\begin{proof}
  Let
  $\hat{I} \defeq \bigcup \setbarTall{\lowerSet{\cup_{j \in J_0}
      O_j}}{J_0 \subseteq_f J, \forall j \in J_0: O_j \in
    I_j}$. First, we show that the set $\hat{I}$ is an ideal:
  \begin{itemize}
  \item $\hat{I}$ is a lower set: Take any $x \in \hat{I}$. Then, for
    some finite set $J_0 \subseteq_f J$ and a collection
    $\setbarNormal{O_j}{j \in J_0}$ satisfying
    $\forall j \in J_0: O_j \in I_j$, we have
    $x \in \lowerSet{\cup_{j \in J_0} O_j}$, which implies that
    $x \subseteq \cup_{j \in J_0} O_j$. Hence, any $y \in \Omega_0$
    satisfying $y \subseteq x$ must also satisfy
    $y \subseteq \cup_{j \in J_0} O_j$, {\ie},
    $y \in \lowerSet{\cup_{j \in J_0} O_j} \subseteq \hat{I}$.

  \item $\hat{I}$ is directed: Let $x,y \in \hat{I}$. Then, for some
    finite sets $J_x, J_y \subseteq_f J$ and for some collections
    $\setbarNormal{O_j}{j \in J_x}$ and
    $\setbarNormal{O'_j}{j \in J_y}$ satisfying
    $\forall j \in J_x: O_j \in I_j$ and
    $\forall j \in J_y: O'_j \in I_j$, we have
    $x \subseteq \cup_{j \in J_x}O_j$ and
    $y \subseteq \cup_{j \in J_y}O'_j$. Let $J_0 \defeq J_x \cup J_y$
    and define:
    \begin{equation*}
      \left\{
      \begin{array}{ll}
        \forall j \in J_0 \setminus J_x:& O_j = \emptyset,\\
        \forall j \in J_0 \setminus J_y:& O'_j = \emptyset.
      \end{array}
      \right.
    \end{equation*}
    As such,
    $x \cup y \subseteq (\cup_{j \in J_x}O_j) \cup (\cup_{j \in
      J_y}O'_j) = \cup_{j \in J_0} (O_j \cup O'_j)$, which implies
    that $x \cup y \in \lowerSet{\cup_{j \in J_0} (O_j \cup
      O'_j)}$. By Proposition~\ref{prop:ideals_closed_finite_union},
    $\forall j \in J_0: O_j \cup O'_j \in I_j$. Hence,
    $x \cup y \in \hat{I}$.
\end{itemize}

The ideal $\hat{I}$ is an upper bound of
$\setbarNormal{I_j}{j \in J}$. This is because, for every $j \in J$
and $O \in I_j$, we have $O \in \lowerSet{O} \subseteq \hat{I}$.

Assume that $\tilde{I} \in {\cal L}$ is another upper bound of
$\setbarNormal{I_j}{j \in J}$. As $\tilde{I}$ is a directed lower set, then
for any finite set $J_0 \subseteq_f J$ and collection
$\setbarNormal{O_j}{j \in J_0}$ satisfying
$\forall j \in J_0: O_j \in I_j$, we must have
$\bigcup_{j \in J_0} O_j \in \tilde{I}$, which implies that
$\lowerSet{\bigcup_{j \in J_0} O_j} \subseteq \tilde{I}$. Therefore,
$\hat{I} \subseteq \tilde{I}$.
\end{proof}

Recall that a lattice is said to be spatial
if it is order isomorphic to $( \tau_{\mathbb{M}}, \subseteq)$ for
some topological space $\mathbb{M} \equiv (M, \tau_{\mathbb{M}})$.

\begin{proposition}
  \label{prop:L_spatial}
  The lattice ${\cal L} = \Idl(\Omega_0)$ is spatial.
\end{proposition}
\begin{proof}
  Follows from the fact that every continuous distributive lattice is
  spatial~\cite[Lemma~7.2.15]{AbramskyJung94-DT}.
\end{proof}

As such, $\Idl(\Omega_0)$ is (order isomorphic to) the lattice of open
subsets of a topological space $\spect{\mathbb{X}}_{\Omega_0}$ which
we regard as a spectral compactification of $\mathbb{X}$:

\begin{definition}[Spectral compactification: $\spect{\mathbb{X}}_{\Omega_0}$]
  Assume that $\mathbb{X} \equiv (X, \tau_{\mathbb{X}})$ is a $T_0$
  topological space and $\Omega_0 \subseteq \tau_{\mathbb{X}}$ is a
  viable base of $\mathbb{X}$. By the spectral compactification of
  $\mathbb{X}$ generated by $\Omega_0$ we mean the topological space:
  \begin{equation*}
  \spect{\mathbb{X}}_{\Omega_0} \equiv (\spect{X}_{\Omega_0},
  \spect{\tau}) \defeq \pt({\cal L}) = \pt(\Idl(\Omega_0)),
  \end{equation*}
  in which $\spect{\tau}$ is the hull-kernel topology. When $\Omega_0$
  is clear from the context, we use the simpler notation
  $\spect{\mathbb{X}}$.
\end{definition}

Recall that, in the hull-kernel topology $\spect{\tau}$ on
$\spect{\mathbb{X}}_{\Omega_0}$, every open set is of the form:
\begin{equation}
  \label{eq:hull_kernel_O_I}
{\cal O}_I \defeq \setbarNormal{y \in \spect{X}_{\Omega_0}}{I \in y},  
\end{equation}
\noindent
in which $I$ ranges over $\Idl(\Omega_0)$.

\begin{proposition}
  \label{prop:O_I_monotone}
  $\forall I_1, I_2 \in \Idl(\Omega_0): I_1 \subseteq I_2 \implies
  {\cal O}_{I_1} \subseteq {\cal O}_{I_2}$.
\end{proposition}

\begin{proof}
  Take any $y \in {\cal O}_{I_1}$. Then,
  by~\eqref{eq:hull_kernel_O_I}, $I_1 \in y$. As $y$ is a filter and
  $I_1 \subseteq I_2$, then $I_2 \in y$, which implies that
  $y \in {\cal O}_{I_2}$. Hence,
  ${\cal O}_{I_1} \subseteq {\cal O}_{I_2}$.
\end{proof}

\begin{proposition}
  \label{prop:O_I_Union}
  Assume that $I \in \Idl(\Omega_0)$ and for a family
  $\setbarNormal{I_j}{j \in J} \subseteq \Idl(\Omega_0)$, we have
  $I = \join_{j \in J} I_j$. Then,
  ${\cal O}_I = \bigcup_{j \in J} {\cal O}_{I_j}$.
\end{proposition}

\begin{proof}
  The fact that
  ${\cal O}_I \supseteq \bigcup_{j \in J} {\cal O}_{I_j}$ follows from
  Proposition~\ref{prop:O_I_monotone}. To prove the $\subseteq$
  direction, for any $y \in {\cal O}_I$:
  \begin{align*}
    y \in {\cal O}_I &\implies I \in y\\
      (\text{since } I = \join_{j \in J}I_j)               &\implies \join_{j \in J}I_j \in y\\
      (\text{$y$ is completely prime})               &\implies \exists j \in J: I_j \in y\\
(\text{by~\eqref{eq:hull_kernel_O_I}})    &\implies y \in {\cal O}_{I_j}.
  \end{align*}
\end{proof}

\begin{corollary}
  \label{cor:base_for_hull_kernel}
  The set $\setbarNormal{{\cal O}_{\lowerSet{W}}}{W \in \Omega_0}$
  forms a base for the hull-kernel topology $\spect{\tau}$ on
  $\spect{\mathbb{X}}_{\Omega_0}$.
\end{corollary}

\begin{proof}
  Follows from Proposition~\ref{prop:O_I_Union}, and the fact that
  $\setbarNormal{\lowerSet{W}}{W \in \Omega_0}$ forms a
  (domain-theoretic) basis for the lattice $\Idl(\Omega_0)$.
\end{proof}

We know that $\spect{\mathbb{X}}_{\Omega_0}$ is a spectral space,
hence, it is core-compact. In what follows, we show that, whenever
$\Omega_0$ separates points, $\spect{\mathbb{X}}_{\Omega_0}$ is indeed
a core-compactification of $\mathbb{X}$. We consider the map
$\iota: \mathbb{X} \to \spect{\mathbb{X}}_{\Omega_0}$ defined by:

  \begin{equation}
    \label{eq:def_embedding}
    \forall x \in X: \quad \iota(x) \defeq \setbarTall{I
      \in \Idl(\Omega_0)}{x \in \cup I}.
  \end{equation}
  
  \begin{proposition}
    The map $\iota$ defined in~\eqref{eq:def_embedding} is
    well-defined, {\ie}, $\forall x \in X: \iota(x)$ is a completely
    prime filter.
  \end{proposition}

  \begin{proof}
    We first prove that, $\forall x \in X: \iota(x)$ is a filter:
    \begin{itemize}
    \item $\iota(x) \neq \emptyset$, because $\Omega_0 \in \iota(x)$.
    \item $\iota(x)$ is an upper set as a consequence of the
      monotonicity of union. To be more precise, assume that
      $I \in \iota(x)$ and $I' \in \Idl(\Omega_0)$ satisfy
      $I \subseteq I'$. Then, $x \in \cup I \subseteq \cup I'$, which
      implies that $x \in \cup I'$. Hence, $I' \in \iota(x)$.

    \item $\iota(x)$ is downward directed: Let
      $I_1, I_2 \in \iota(x)$. Then, for some $O_1 \in I_1$ and
      $O_2 \in I_2$, we must have $x \in O_1$ and $x \in O_2$. Thus,
      $x \in O_1 \cap O_2$. By
      Proposition~\ref{prop:intersection_ideals_open_sets},
      $O_1 \cap O_2 \in I_1 \cap I_2$. By
      Proposition~\ref{prop:L_complet_lattice}, we must have
      $I_1 \wedge I_2 \in \iota(x)$.
    \end{itemize}
    
    Next, we prove that $\iota(x)$ is completely prime. Assume that
    $A \defeq \setbarNormal{I_j}{j \in J} \subseteq \Idl(\Omega_0)$
    and $\join A \in \iota(x)$. Note that $A \neq \emptyset$,
    otherwise we would have:
    \begin{equation*}
    \join A = \join \emptyset = \bot_{\Idl(\Omega_0)} = \lowerSet{\emptyset} \notin \iota(x).      
  \end{equation*}
  Hence, $A$ must be non-empty. By~\eqref{eq:lub_explicit}, for some
  finite set $J_0 \subseteq_f J$ and a collection
  $\setbarNormal{O_j}{j \in J_0 \text{ and } O_j \in I_j}$, we must
  have $x \in \cup_{j \in J_0} O_j$. Thus,
  $\exists j_0 \in J_0: x \in O_{j_0} \subseteq \cup I_{j_0}$, which
  implies that $I_{j_0} \in \iota(x)$.
  \end{proof}

  \begin{proposition}
    \label{prop:iota_cont}
    The map $\iota: \mathbb{X} \to \spect{\mathbb{X}}_{\Omega_0}$ is
    continuous and $\forall I \in \Idl(\Omega_0): \iota^{-1}({\cal
      O}_I) = \cup I$.
  \end{proposition}

  \begin{proof}
    For all $I \in \Idl(\Omega_0)$ and $x \in X$, we have:
    $ x \in \cup I \iff I \in \iota(x) \iff \iota(x) \in {\cal O}_I$.
  \end{proof}
  
  Using Proposition~\ref{prop:iota_cont}, for the case of
  $\mathbb{Y} = \spect{\mathbb{X}}_{\Omega_0}$, we obtain the following
  alternative formulation of~\eqref{eq:f_*_embedding}:
  
  \begin{proposition}
    For any $f \in [\mathbb{X} \to \mathbb{D}]$, we have:
    \begin{equation*}
      %\label{eq:f_*_ideal_formulation}
      \forall \hat{x} \in \spect{X}_{\Omega_0}: \quad
 f_*(\hat{x}) = \join
    \setbarTall{\meet f(\cup I)}{ I \in \hat{x}}.
  \end{equation*}
\end{proposition}

\begin{proof}
  Every open set $\hat{O}$ in the hull-kernel topology $\hat{\tau}$ is
  of the form ${\cal O}_I$ (as in~\eqref{eq:hull_kernel_O_I}) for some
  ideal $I \in \Idl(\Omega_0)$ and we know that
  $\hat{x} \in {\cal O}_I \iff I \in \hat{x}$. Furthermore, by
  Proposition~\ref{prop:iota_cont},
  $\forall I \in \Idl(\Omega_0): \iota^{-1}({\cal O}_I) = \cup I$.
\end{proof}

To show that the map
$\iota: \mathbb{X} \to \spect{\mathbb{X}}_{\Omega_0}$ is a
quasi-embedding, we must use the assumption that $\Omega_0$ is a
base. It is straightforward to verify that, for any open set
$W \in \tau_{\mathbb{X}}$ ({\ie}, not necessarily in $\Omega_0$) we
have
$\lowerSet{W} \defeq \setbarNormal{U \in \Omega_0}{U \subseteq W} \in
\Idl(\Omega_0)$. With that in mind:

\begin{proposition}
  \label{prop:W_tau_X}
  For every open set $W \in \tau_{\mathbb{X}}$ and $x \in X$, we have $x \in W \iff \lowerSet{W} \in \iota(x)$.
\end{proposition}

\begin{proof}
  The $(\Leftarrow)$ direction is straightforward. For the
  $(\Rightarrow)$ implication, since $\Omega_0$ is a base, we have:
  \begin{equation*}
    x \in W \implies \exists U \in \Omega_0: x \in U \subseteq W
    \implies x \in \bigcup \lowerSet{W} \implies \lowerSet{W} \in \iota(x).
  \end{equation*}
\end{proof}

  \begin{proposition}
    \label{prop:iota_forward_W}
    For every open set $W \in \tau_{\mathbb{X}}$, we have
    $\iota(W) = {\cal O}_{\lowerSet{W}} \cap \iota(X)$.
  \end{proposition}

  \begin{proof}
    By using Proposition~\ref{prop:W_tau_X}, we obtain $\forall x \in X: x \in W \iff \lowerSet{W} \in \iota(x) \iff \iota(x) \in {\cal
      O}_{\lowerSet{W}}$.
  \end{proof}

  \begin{corollary}
    The map $\iota: \mathbb{X} \to \spect{\mathbb{X}}_{\Omega_0}$ is a
    topological quasi-embedding.
  \end{corollary}

  \begin{proof}
    Follows from Propositions~\ref{prop:iota_cont}, and
    \ref{prop:iota_forward_W}.
  \end{proof}

  \begin{lemma}
    \label{lemma:iota_dense}
    The quasi-embedding
    $\iota: \mathbb{X} \to \spect{\mathbb{X}}_{\Omega_0}$ is dense.
  \end{lemma}

  \begin{proof}
    From Propositions~\ref{prop:iota_cont},
    and~\ref{prop:iota_forward_W}, we obtain:
    $\forall I \in \Idl(\Omega_0): {\cal O}_I \cap \iota(X) =
    \iota(\cup I)$.
  \end{proof}
  
  So far, we have not used separation of points:
  \begin{proposition}
    \label{prop:iota_injective}
    Whenever $\Omega_0$ separates points, the map
    $\iota: \mathbb{X} \to \spect{\mathbb{X}}_{\Omega_0}$ is injective.
  \end{proposition}

  \begin{proof}
    Let $x,y \in X$, and assume that $x \neq y$. Since $\Omega_0$
    separates points, there exists an open set $W \in \Omega_0$ which
    includes one point but not the other. Without loss of generality,
    we assume that $x \in W$ and $y \notin W$. Then
    $\lowerSet{W} \in \iota(x)$ but $\lowerSet{W} \notin \iota(y)$.
  \end{proof}

  To summarize, we have proven that:
  \begin{theorem}[Core-compactification]
    Assume that $\mathbb{X}$ is a $T_0$ topological space. If
    $\Omega_0$ is a viable base of $\mathbb{X}$, then the spectral
    space $\spect{\mathbb{X}}_{\Omega_0}$ is a core-compactification
    of $\mathbb{X}$.
  \end{theorem}
  
  \begin{example}[rational upper limit topology]
    \label{example:rational_upper_limit}
    Consider the set
    $B_{(\Q]} \defeq \setbarNormal{(a,b]}{ a, b \in \Q}$ of left
    half-open intervals with rational end-points. The collection
    $B_{(\Q]}$ forms a base for what we refer to as the \emph{rational
      upper limit topology}. Let $\R_{(\Q]} \equiv( \R, \tau_{(\Q]})$
    denote the topological space with $\R$ as the carrier set endowed
    with the rational upper limit topology $\tau_{(\Q]}$. As for a
    viable $\Omega_0$, an immediate option is $\tau_{(\Q]}$. Yet,
    considering Corollary~\ref{cor:base_for_hull_kernel}, it will not
    lead to a second-countable $\spect{\mathbb{X}}_{\Omega_0}$.

    Instead, we can take $\Omega_0$ to consist of all the finite
    unions of elements of $B_{(\Q]}$. This is a countable set which can
    be effectively enumerated. By
    Corollary~\ref{cor:base_for_hull_kernel}, the spectral space
    $\spect{\mathbb{X}}_{\Omega_0}$ must also be second-countable.
  \end{example}
  
  The rational upper limit topology is used
  in~\cite{Edalat_Farjudian_Li:Temporal_Discretization:2023}
  for solution of IVPs with temporal
  discretization. In~\cite{Edalat_Farjudian_Li:Temporal_Discretization:2023},
  the domain $[\mathbb{Y} \to \mathbb{D}]$ is constructed by rounded
  ideal completion of a suitable abstract basis of step functions. In
  this article, we work directly on $\mathbb{X}$. As we will see
  (Theorem~\ref{thm:two_approaches_equiv}) the two approaches lead to
  equivalent outcomes.

  \begin{remark}
    We have presented the construction of the spectral
    compactification by first considering the ideal completion
    $\Idl(\Omega_0)$ of $\Omega_0$, and then applying the $\pt$
    functor. The same result can be obtained by directly considering
    the prime filters of
    $\Omega_0$~\cite[Proposition~7.2.23]{AbramskyJung94-DT}. We
    opted for the two-step construction because, in
    Section~\ref{sec:domain_funs}, we will need to refer to some
    properties of $\Idl(\Omega_0)$.
  \end{remark}

%%%%%%%%%%%%%%%%%%%%%%%%%%%%%%%%%%%%%%%%%%%%%%%%%%%%%%%%%  
\section{Continuous Domain of Functions}
\label{sec:domain_funs}

The aim here is to present a framework for computation over function
spaces. Starting from a topological space $\mathbb{X}$ and a
bc-domain $\mathbb{D}$, if $\mathbb{X}$ is core-compact, then
$[\mathbb{X} \to \mathbb{D}]$ is a bc-domain, and all that remains
is to determine whether $[\mathbb{X} \to \mathbb{D}]$ admits a
suitable effective structure. If it does not, or if $\mathbb{X}$ is
not core-compact to begin with, then our aim is to consider the
substitute bc-domain
$[\spect{\mathbb{X}}_{\Omega_0} \to \mathbb{D}]$, for a viable base
$\Omega_0$ of $\mathbb{X}$.

\begin{theorem}
  Assume that $\mathbb{X} \equiv (X, \tau_{\mathbb{X}})$ is a
  topological space and $\Omega_0 \subseteq \tau_{\mathbb{X}}$ is a
  viable base of $\mathbb{X}$. Let
  $\mathbb{D} \equiv (D, \sqsubseteq)$ be a bc-domain and assume
  that $D_0 \subseteq D$ is a basis for $\mathbb{D}$. Then,
  $[\spect{\mathbb{X}}_{\Omega_0} \to \mathbb{D}]$ is a bc-domain
  with a basis $\hat{\mathbb{B}}$ of step-functions of the form:
  \begin{equation}
    \label{eq:basis_B}
    \hat{\mathbb{B}} = \setbarTall{\join_{i \in I} b_i \chi_{{\cal
          O}_{\lowerSet{W_i}}}}{ \text{$I$ is finite, $\setbarTall{b_i \chi_{{\cal
          O}_{\lowerSet{W_i}}}}{i \in I}$ is
        consistent}, \forall i \in I: W_i \in \Omega_0, b_i \in D_0}.
  \end{equation}
\end{theorem}

\begin{proof}
  The fact that $[\spect{\mathbb{X}}_{\Omega_0} \to \mathbb{D}]$ is a
  bc-domain follows from
  Theorem~\ref{thm:core_compact_bc_domain}. By
  Corollary~\ref{cor:base_for_hull_kernel}, the set
  $\setbarNormal{{\cal O}_{\lowerSet{W}}}{W \in \Omega_0}$ forms a
  base for the hull-kernel topology $\spect{\tau}$ on
  $\spect{\mathbb{X}}_{\Omega_0}$. Hence, from
  Lemma~\ref{lemma:generalized_erker_sup_step_fun}, we deduce:
  \begin{equation*}
  \forall f \in [\spect{\mathbb{X}}_{\Omega_0} \to \mathbb{D}]:\quad  f
  = \join \setbarTall{b \chi_{{\cal
          O}_{\lowerSet{W}}} }{\lowerSet{W} \ll
    f^{-1}(\wayaboves{b}), W \in \Omega_0, b \in D_0},
  \end{equation*}
  which implies that $\hat{\mathbb{B}}$ is indeed a basis for
  $[\spect{\mathbb{X}}_{\Omega_0} \to \mathbb{D}]$.
\end{proof}

\begin{corollary}
  If $\Omega_0$ is countable and $\mathbb{D}$ is $\omega$-continuous,
  then $[\spect{\mathbb{X}}_{\Omega_0} \to \mathbb{D}]$ is also $\omega$-continuous.
\end{corollary}

\begin{proof}
  If $\mathbb{D}$ is $\omega$-continuous, then $D_0$ can be chosen to
  be countable. By~\eqref{eq:basis_B}, $\hat{\mathbb{B}}$ is a
  countable basis for $[\spect{\mathbb{X}}_{\Omega_0} \to \mathbb{D}]$.
\end{proof}

Regarding the way-below relation over
$[\spect{\mathbb{X}}_{\Omega_0} \to \mathbb{D}]$, we can derive a
useful formulation for the basis elements, {\ie}, step functions of
the form~\eqref{eq:basis_B}. Recall that a core-compact space
$\mathbb{Y} \equiv (Y, \tau_{\mathbb{Y}})$ is called \emph{stable} if
$U \ll V$ and $U \ll V'$ imply $U \ll V \cap V'$, for all
$U, V, V' \in \tau_{\mathbb{Y}}$. When $\mathbb{Y}$ is stable, based
on~\cite[Lemma~1 and Proposition~5]{Erker_et_al:way_below:1998},
we know that, for any step function
$\join_{i \in I} b_i \chi_{O_i} \in [\mathbb{Y} \to \mathbb{D}]$ and
any arbitrary function $g \in [\mathbb{Y} \to \mathbb{D}]$:
  \begin{equation}
    \label{eq:step_fun_way_below_formulation}
  \join_{i \in I} b_i \chi_{O_i} \ll  g  \iff \forall i \in I: O_i \ll
  g^{-1}(\wayaboves{b_i}). 
\end{equation}
We point out that stability is used for the $(\Rightarrow)$ direction,
{\ie}, for the $(\Leftarrow)$ implication to hold, it suffices for
$\mathbb{Y}$ to be
core-compact~\cite[Lemma~1]{Erker_et_al:way_below:1998}.

Let us now take $\mathbb{Y} = \spect{\mathbb{X}}_{\Omega_0}$ and
recall that every spectral space is stable. Hence, we can
use~\eqref{eq:step_fun_way_below_formulation}. Assume that, for some
finite index set $J$, we have
$g \defeq \join_{j \in J} b' \chi_{{\cal O}_{\lowerSet{W'_j}}}$. It is
easy to show that for all $b \in D$, there exists
a finite index set $J_b$ such that
$g^{-1}(\wayaboves b) = \bigcup_{k \in J_b} {\cal
  O}_{\lowerSet{W''_k}}$, in which each $W''_k$ is formed by
intersecting some open sets of the form $W'_j$, {\ie}, there exists a
$J_k \subseteq J$ such that $W''_k = \bigcap_{j \in J_k}W'_j$. By
Proposition~\ref{prop:O_I_Union}, we have:
\begin{equation}
  \label{eq:O_Union_Ideal}
  \bigcup_{k \in J_b} {\cal O}_{\lowerSet{W''_k}} = {\cal O}_{\join_{k
      \in J_b} \lowerSet{W''_k}} ={\cal O}_{\lowerSet{\bigcup_{k \in J_b}W''_k}}.
\end{equation}
\noindent
According to Proposition~\ref{prop:L_spatial}, the complete lattice
${\cal L} = \Idl(\Omega_0)$ is spatial and we have:
\begin{equation}
  \label{eq:Omega_Idl_equiv}
 \Omega(\spect{\mathbb{X}}_{\Omega_0}) \cong
\Omega(\pt(\Idl(\Omega_0))) \cong \Idl(\Omega_0).
\end{equation}
For instance, we have:

\begin{eqnarray*}
  && {\cal O}_{\lowerSet{W}} \ll {\cal
                         O}_{\lowerSet{W'}} \quad (\text{in
                         $(\spect{\tau}, \subseteq)$}) \\
  (\text{by~\ref{eq:Omega_Idl_equiv}}) &\iff& \lowerSet{W} \ll
                                         \lowerSet{W'} \quad (\text{in
                                         $(\Idl(\Omega_0),
                                         \subseteq)$}) \\
  (\text{by~\cite[Proposition~2.2.22]{AbramskyJung94-DT}}) &\iff&
                                                                  W
                                                                  \subseteq
                                                                  W'
                                                                  \quad
                                                                  (\text{in
                                                                  $\Omega_0$}).
\end{eqnarray*}
\noindent
In fact, we have:
\begin{lemma}
  \label{lemma:step_functions_way_below_Spectral}
  The way-below relation on step functions of~\eqref{eq:basis_B} can
  be expressed as:
  \begin{equation*}
  \join_{i \in I} b_i \chi_{{\cal O}_{\lowerSet{W_i}}} \ll  \join_{j
    \in J} b'_j \chi_{{\cal
      O}_{\lowerSet{W'_j}}}  \iff   \forall i \in I: W_i
  \subseteq U_i,
\end{equation*}
in which $U_i \in \Omega_0$ satisfies ${\cal O}_{\lowerSet{U_i}} =
\left( \join_{j
    \in J} b'_j \chi_{{\cal
      O}_{\lowerSet{W'_j}}} \right)^{-1}(\wayaboves{b_i})$.
\end{lemma}

\begin{proof}
  This follows from~\eqref{eq:step_fun_way_below_formulation},
  \eqref{eq:O_Union_Ideal}, and the fact that $\Omega_0$ is closed
  under finite unions and finite intersections.
\end{proof}

%%%%%%%%%%%%%%%%%%%%%%%%%%%%%%%%%%%%%%%%%%%%%%%%%%%%%%%%%%%%%%%
\subsection{Construction Using Abstract Bases}

Lemma~\ref{lemma:step_functions_way_below_Spectral} suggests an
alternative approach to obtaining a domain of functions based on \emph{abstract bases} without
referring to Stone duality:

\begin{definition}[Abstract basis]
  \label{def:abstract_basis}
  A pair $(B, \precAbsBas)$ consisting of a set $B$ and a binary relation
  $\precAbsBas\ \subseteq B \times B$ is said to be an abstract basis if the
  relation $\precAbsBas$ is transitive and satisfies the following
  interpolation property:
  \begin{itemize}
  \item For every finite subset $A \subseteq_f B$ and any element
    $x \in B: A \precAbsBas x \implies \exists y \in B: A \precAbsBas
    y \precAbsBas x$.
  \end{itemize}
  Here, by $A \precAbsBas x$ we mean $\forall a \in A: a \precAbsBas x$. 
\end{definition}

In this approach, we work directly with step functions in $[\mathbb{X}
\to \mathbb{D}]$. Specifically, we consider:

\begin{equation}
  \label{eq:def_B_abs}
  \mathbb{B}_{\mathrm{abs}} \defeq \setbarTall{f: X \to D}{f = \join_{i \in I} b_i
    \chi_{O_i}, I \text{ is finite, $\setbarTall{b_i
    \chi_{O_i}}{i \in I}$  is consistent}, \forall i \in I: O_i \in \Omega_0
    \text{ and }
    b_i \in D_0 }.
\end{equation}
\noindent
As for the binary relation $\precAbsBas$, considering
Lemma~\ref{lemma:step_functions_way_below_Spectral}, we define:
\begin{equation}
  \label{eq:def_precAbsBas}
  \join_{i \in I} b_i \chi_{O_i} \precAbsBas  \join_{j \in J} b'_j
  \chi_{O'_j} \iff \forall i \in I: O_i \subseteq (\join_{j \in J} b'_j
  \chi_{O'_j})^{-1}(\wayaboves{b_i}). 
\end{equation}
This is indeed the approach taken
in~\cite{Edalat_Farjudian_Li:Temporal_Discretization:2023}, where
$\mathbb{X}$ is the real line endowed with the rational upper limit
topology (Example~\ref{example:rational_upper_limit}) and $\mathbb{D}$
is the interval domain $\intvaldom[\R^n\lift]$
of~\eqref{eq:interval_domain_IR_n}. In~\cite{Edalat_Farjudian_Li:Temporal_Discretization:2023},
a domain ${\cal W}$ is constructed as the rounded ideal completion of
$(\mathbb{B}_{\mathrm{abs}}, \precAbsBas)$ and the Galois connection
of Theorem~\ref{thm:Galois_connection_X_to_D} is obtained with
$[\mathbb{Y} \to \mathbb{D}]$ replaced by ${\cal W}$:
  \begin{equation*}
    %\label{eq:Galois_Connection_QEmbed}
    \begin{tikzcd}[column sep = large]
      [\mathbb{X} \to \mathbb{D}] \arrow[r, "(\cdot)_*", yshift =
      1.1ex] & {\cal W}  \arrow[l, "(\cdot)^*", "\top"', yshift = -1.1ex].
    \end{tikzcd}
  \end{equation*}
  The aim
  in~\cite{Edalat_Farjudian_Li:Temporal_Discretization:2023}
  has been solution of IVPs with temporal discretization. In that
  context, the computation is carried out over the (non-continuous)
  dcpo $[\mathbb{X} \to \mathbb{D}]$, while the theoretical
  analyses (including computable analysis) is carried out over the
  continuous domain ${\cal W}$. The Galois connection provides the
  bridge between the two.
  
  We expect this to be the general rule. When $\mathbb{X}$ is not
  core-compact, the non-continuous dcpo $[\mathbb{X} \to \mathbb{D}]$
  is still useful for implementation of algorithms. But it cannot be
  used for computable analysis. With careful choice of $\Omega_0$, it
  is possible to endow
  $[\spect{\mathbb{X}}_{\Omega_0} \to \mathbb{D}]$ with an effective
  structure which makes it suitable for computable analysis of the
  relevant problem.

To demonstrate that the two approaches are equivalent, we show that
the (abstract) bases $(\hat{\mathbb{B}}, \ll)$ of~\eqref{eq:basis_B}
and $(\mathbb{B}_{\mathrm{abs}}, \precAbsBas)$ of~\eqref{eq:def_B_abs}
are isomorphic. We define a map
$\gamma : \hat{\mathbb{B}} \to \mathbb{B}_{\mathrm{abs}}$ by
$\forall f \in \hat{\mathbb{B}}: \gamma(f) \defeq f \circ \iota =
f^*$. It is straightforward to verify that, for any
$\join_{i \in I} b_i \chi_{W_i} \in \mathbb{B}_{\mathrm{abs}}$, we
have
$\join_{i \in I} b_i \chi_{W_i} = \gamma \left( \join_{i \in I} b_i
  \chi_{{\cal O}_{\lowerSet{W_i}}} \right)$. It is also easy to verify
that $\gamma : \hat{\mathbb{B}} \to \mathbb{B}_{\mathrm{abs}}$ and its
inverse $\gamma^{-1}: \mathbb{B}_{\mathrm{abs}} \to \hat{\mathbb{B}}$
are monotonic and bijective. We must show, however, that $\gamma$
preserves the way-below relations.

  \begin{lemma}
    $\forall f, g \in \hat{\mathbb{B}}: \quad f \ll g \iff \gamma(f)
    \precAbsBas \gamma(g)$.
  \end{lemma}

  \begin{proof}
    This is almost immediate from
    Lemma~\ref{lemma:step_functions_way_below_Spectral}
    and~\eqref{eq:def_precAbsBas}, except that we must show that the
    relations hold regardless of how the step-functions are
    represented. In other words, whenever
    $\join_{k \in K} c_k \chi_{U_k} = \gamma\left( \join_{i \in I} b_i
      \chi_{{\cal O}_{\lowerSet{W_i}}} \right)$ and
    $\join_{\ell \in L} c'_{\ell} \chi_{U'_{\ell}} = \gamma\left(
      \join_{j \in J} b'_j \chi_{{\cal O}_{\lowerSet{W'_j}}} \right)$:
    \begin{equation}
      \label{eq:way_below_independent}
    \join_{i \in I} b_i \chi_{{\cal O}_{\lowerSet{W_i}}} \ll  \join_{j
      \in J} b'_j \chi_{{\cal
        O}_{\lowerSet{W'_j}}}  \iff   \join_{k \in K} c_k \chi_{U_k} \precAbsBas  \join_{\ell \in L} c'_{\ell}
    \chi_{U'_{\ell}}.
  \end{equation}
  This follows from the fact that $\gamma$ is monotone and
  bijective. For instance, to prove the $(\Leftarrow)$ implication
  in~\eqref{eq:way_below_independent}, assume that
  $\join_{k \in K} c_k \chi_{U_k} \precAbsBas \join_{\ell \in L}
  c'_{\ell} \chi_{U'_{\ell}}$. Then, we have:
  \begin{equation*}
    \left\{
      \begin{array}{l}
        \gamma\left( \join_{k \in K} c_k
      \chi_{{\cal O}_{\lowerSet{U_k}}} \right) = 
      \join_{k \in K} c_k \chi_{U_k} = \gamma\left( \join_{i \in I} b_i
        \chi_{{\cal O}_{\lowerSet{W_i}}} \right), \\
        \gamma\left( \join_{\ell \in L} c'_{\ell} \chi_{{\cal
        O}_{\lowerSet{U'_{\ell}}}} \right) =
      \join_{\ell \in L} c'_{\ell} \chi_{U'_{\ell}} = \gamma\left(
      \join_{j \in J} b'_j \chi_{{\cal O}_{\lowerSet{W'_j}}} \right).
    \end{array}
    \right.
  \end{equation*}
  As $\gamma$ is bijective, we must have
  $\join_{k \in K} c_k \chi_{{\cal O}_{\lowerSet{U_k}}} = \join_{i \in
    I} b_i \chi_{{\cal O}_{\lowerSet{W_i}}}$ and
  $\join_{\ell \in L} c'_{\ell} \chi_{{\cal O}_{\lowerSet{U'_{\ell}}}}
  = \join_{j \in J} b'_j \chi_{{\cal O}_{\lowerSet{W'_j}}}$. The
  result now follows from
  Lemma~\ref{lemma:step_functions_way_below_Spectral}
  and~\eqref{eq:def_precAbsBas}.
\end{proof}

The approach based on abstract bases also provides us with the same
continuous domain of functions:

\begin{theorem}
  \label{thm:two_approaches_equiv}
  Assume that the domain ${\cal W}$ is the rounded ideal completion of
  $(\mathbb{B}_{\mathrm{abs}}, \precAbsBas)$. Then
  ${\cal W} \cong [\spect{\mathbb{X}}_{\Omega_0} \to \mathbb{D}]$.
\end{theorem}

\begin{proof}
  Follows from the fact that the (abstract) bases
  $(\hat{\mathbb{B}}, \ll)$ and
  $(\mathbb{B}_{\mathrm{abs}}, \precAbsBas)$ are isomorphic.
\end{proof}

%%%%%%%%%%%%%%%%%%%%%%%%%%%%%%%%%%%%%%%%%%%%%%%%%%%%%%%%%%%%%%%%
\section{Concluding Remarks}
\label{sec:Concluding_Remarks}

We have investigated the basic properties of a computational framework
for function spaces over topological spaces that are not
core-compact. To that end, we considered spectral compactification of
a given space and presented the construction directly without
referring to quasi-proximities. In our framework, for a space
$\mathbb{Y}$ to be a compactification of a space $\mathbb{X}$, there
must exist a \emph{dense} embedding $\iota: \mathbb{X} \to \mathbb{Y}$
(Definition~\ref{def:core_compactification}). By using
quasi-proximities, one may require the existence of a so-called
\emph{basis} embedding $\iota: \mathbb{X} \to \mathbb{Y}$, which is a
stronger condition~\cite{Smyth:Stable_Compactification:1992}. It will
be worth investigating how our results can be strengthened with the
basis embedding requirement.

In our framework, computability is analyzed in the continuous domain
$[\spect{\mathbb{X}}_{\Omega_0} \to \mathbb{D}]$. In Type-II Theory of
Effectivity (TTE)~\cite{Weihrauch2000:book}, computability is analyzed
via admissible representations of the function space
$\mathbb{D}^{\mathbb{X}}$. Part of our future work is the
investigation of how the two approaches are related.

Regarding applications of the framework, it has provided a suitable
semantic model for solution of IVPs. To be more precise,
in~\cite{Edalat_Farjudian_Li:Temporal_Discretization:2023}, we
constructed a domain using abstract bases for solution of IVPs with
temporal discretization. In Theorem~\ref{thm:two_approaches_equiv} of
the current article, we showed that the same domain (up to
isomorphism) can be obtained using the construction of the current
article. We believe that spectral compactification will be useful in
domain theoretic solution of partial differential equations (PDEs) as
well.

Spectral compactification provides another angle on the construction
obtained via abstract bases
in~\cite{Edalat_Farjudian_Li:Temporal_Discretization:2023}. Even
though, in practice, the two approaches lead to isomorphic function
spaces, we believe that the construction of the current article has
some theoretical advantages. Compactification is a central topic in
topology, and as we have pointed out earlier, our construction can be
obtained as a special case of Smyth's stable compactification by
considering fine
quasi-proximities~\cite{Smyth:Stable_Compactification:1992}. It will
be interesting to see if any concrete application necessitates
consideration of stable compactifications obtained via
quasi-proximities other than the fine one, in the same way that
solution of IVPs with temporal discretization led us to the
investigation of spectral compactification.

% \bibliographystyle{./entics}
% \bibliography{Biblio}

\begin{thebibliography}{10}
\providecommand{\url}[1]{\texttt{#1}}
\providecommand{\urlprefix}{ }
\providecommand{\eprint}[2][]{\url{#2}}

\bibitem{AbramskyJung94-DT}
Abramsky, S. and A.~Jung, \emph{Domain theory}, in: S.~Abramsky, D.~M. Gabbay
  and T.~S.~E. Maibaum, editors, \emph{Handbook of Logic in Computer Science},
  volume~3, pages 1--168, Clarendon Press, Oxford (1994).
\newline\urlprefix\url{https://doi.org/10.1093/oso/9780198537625.003.0001}

\bibitem{Bilokon_Edalat:Domain_Brownian:2017}
Bilokon, P. and A.~Edalat, \emph{A domain-theoretic approach to {Brownian}
  motion and general continuous stochastic processes}, Theoretical Computer
  Science \textbf{691}, pages 10--26 (2017).
\newline\urlprefix\url{https://doi.org/10.1016/j.tcs.2017.07.016}

\bibitem{DiGianantonio_Edalat_Gutin-Automatic_Differentiation:2023}
{Di Gianantonio}, P., A.~Edalat and R.~Gutin, \emph{A language for evaluating
  derivatives of functionals using automatic differentiation}, Electronic Notes
  in Theoretical Informatics and Computer Science \textbf{3} (2023).
  Proceedings of MFPS 2023.
\newline\urlprefix\url{https://doi.org/10.46298/entics.12303}

\bibitem{Edalat95:DT-fractals}
Edalat, A., \emph{Dynamical systems, measures and fractals via domain theory},
  {I}nformation and {C}omputation \textbf{120}, pages 32--48 (1995).
\newline\urlprefix\url{https://doi.org/10.1006/inco.1995.1096}

\bibitem{Edalat:Domains_Physics:1997}
Edalat, A., \emph{Domains for computation in mathematics, physics and exact
  real arithmetic}, Bull. Symbolic Logic \textbf{3}, pages 401--452 (1997).
\newline\urlprefix\url{https://doi.org/10.2307/421098}

\bibitem{Edalat_Farjudian_Li:Temporal_Discretization:2023}
Edalat, A., A.~Farjudian and Y.~Li, \emph{Recursive solution of initial value
  problems with temporal discretization}, Theoretical Computer Science page
  114221 (2023).
\newline\urlprefix\url{https://doi.org/10.1016/j.tcs.2023.114221}

\bibitem{Edalat_Farjudian_Mohammadian_Pattinson:2nd_Order_Euler:2020:Conf}
Edalat, A., A.~Farjudian, M.~Mohammadian and D.~Pattinson, \emph{Domain
  theoretic second-order {Euler's} method for solving initial value problems},
  Electr. Notes in Theoret. Comp. Sci. \textbf{352}, pages 105--128 (2020). The
  36th Mathematical Foundations of Programming Semantics Conference, 2020 (MFPS
  2020), Paris, France.
\newline\urlprefix\url{https://doi.org/10.1016/j.entcs.2020.09.006}

\bibitem{Edalat_Lieutier:Domain_Calculus_One_Var:MSCS:2004}
Edalat, A. and A.~Lieutier, \emph{Domain theory and differential calculus
  (functions of one variable)}, Mathematical Structures in Computer Science
  \textbf{14}, pages 771--802 (2004).
\newline\urlprefix\url{https://doi.org/10.1017/S0960129504004359}

\bibitem{Edalat_Pattinson:Hybrid:2007}
Edalat, A. and D.~Pattinson, \emph{Denotational semantics of hybrid automata},
  The Journal of Logic and Algebraic Programming \textbf{73}, pages 3--21
  (2007).
\newline\urlprefix\url{https://doi.org/10.1016/j.jlap.2007.01.002}

\bibitem{Edalat_Pattinson2007-LMS_Picard}
Edalat, A. and D.~Pattinson, \emph{A domain-theoretic account of {Picard's}
  theorem}, LMS Journal of Computation and Mathematics \textbf{10}, pages
  83--118 (2007).
\newline\urlprefix\url{https://doi.org/10.1112/S1461157000001315}

\bibitem{Erker_et_al:way_below:1998}
Erker, T., M.~H. Escardó and K.~Keimel, \emph{The way-below relation of
  function spaces over semantic domains}, Topology and its Applications
  \textbf{89}, pages 61--74 (1998).
\newline\urlprefix\url{https://doi.org/10.1016/S0166-8641(97)00226-5}

\bibitem{Escardo96-tcs}
Escard{\'o}, M.~H., \emph{{PCF} extended with real numbers}, {T}heoretical
  {C}omputer {S}cience \textbf{162}, pages 79--115 (1996).
\newline\urlprefix\url{https://doi.org/10.1016/0304-3975(95)00250-2}

\bibitem{Farjudian:Shrad:2007}
Farjudian, A., \emph{Shrad: A language for sequential real number computation},
  Theory Comput. Syst. \textbf{41}, pages 49--105 (2007).
\newline\urlprefix\url{https://doi.org/10.1007/s00224-006-1339-2}

\bibitem{FarjudianKonecny2008:wollic-lnai}
Farjudian, A. and M.~{Kone\v{c}n{\'y}}, \emph{Time complexity and convergence
  analysis of domain theoretic {Picard} method}, in: W.~Hodges and
  R.~de~Queiroz, editors, \emph{Proceedings of {WoLLIC '08}}, volume 5110 of
  \emph{Lecture Notes in Artificial Intelligence}, pages 149--163, Springer
  (2008).
\newline\urlprefix\url{https://doi.org/10.1007/978-3-540-69937-8_14}

\bibitem{Farjudian_Moggi:Robustness_Scott_Continuity_Computability:2023}
Farjudian, A. and E.~Moggi, \emph{Robustness, {Scott} continuity, and
  computability}, Mathematical Structures in Computer Science \textbf{33},
  pages 536--572 (2023).
\newline\urlprefix\url{https://doi.org/10.1017/S0960129523000233}

\bibitem{Gierz-ContinuousLattices-2003}
Gierz, G., K.~H. Hofmann, K.~Keimel, J.~D. Lawson, M.~W. Mislove and D.~S.
  Scott, \emph{Continuous Lattices and Domains}, volume~93 of \emph{Encycloedia
  of Mathematics and its Applications}, Cambridge University Press (2003).
\newline\urlprefix\url{https://doi.org/10.1017/CBO9780511542725}

\bibitem{Goubault-Larrecq:Non_Hausdorff_topology:2013}
Goubault-Larrecq, J., \emph{Non-Hausdorff topology and domain theory},
  Cambridge University Press (2013).
\newline\urlprefix\url{https://doi.org/10.1017/CBO9781139524438}

\bibitem{Mislove-Topology_DT_TCS:1998}
Mislove, M.~W., \emph{Topology, domain theory and theoretical computer
  science}, Topology and its Applications \textbf{89}, pages 3--59 (1998).
\newline\urlprefix\url{https://doi.org/10.1016/S0166-8641(97)00222-8}

\bibitem{Moggi_Farjudian_Duracz_Taha:Reachability_Hybrid:2018}
Moggi, E., A.~Farjudian, A.~Duracz and W.~Taha, \emph{Safe \& robust
  reachability analysis of hybrid systems}, Theoretical Computer Science
  \textbf{747}, pages 75--99 (2018).
\newline\urlprefix\url{https://doi.org/10.1016/j.tcs.2018.06.020}

\bibitem{Munkres:Topology:2000}
Munkres, J.~R., \emph{Topology}, Prentice Hall, Upper Saddle River, 2nd edition
  (2000).

\bibitem{Smyth:Stable_Compactification:1992}
Smyth, M.~B., \emph{Stable compactification {I}}, Journal of the London
  Mathematical Society \textbf{s2-45}, pages 321--340 (1992).
\newline\urlprefix\url{https://doi.org/10.1112/jlms/s2-45.2.321}

\bibitem{Weihrauch2000:book}
Weihrauch, K., \emph{Computable Analysis, An Introduction}, Springer (2000).
\newline\urlprefix\url{https://doi.org/10.1007/978-3-642-56999-9}

\bibitem{Zhou_Shaikh_Li_Farjudian:Robust_NN:MSCS:2023}
Zhou, C., R.~A. Shaikh, Y.~Li and A.~Farjudian, \emph{A domain-theoretic
  framework for robustness analysis of neural networks}, Mathematical
  Structures in Computer Science \textbf{33}, pages 68--105 (2023).
\newline\urlprefix\url{https://doi.org/10.1017/S0960129523000142}

\end{thebibliography}

% Bibliography inlined for arXive. Taken from the .bbl file.

%\appendix

%%%%%%%%%%%%%%%%%%%%%%%%%%%%%%%%%%%%%%%%%%%%%%%%%%%%%%%%%%%%%%%%%%%%%%%%%%%%%

% Any appendices should be included after the references, as is done here. The appendix starts with the command \verb+\appendix+, after which sections can be included; these are lettered, rather than numbered. 
% \section*{Another appendix}
% One can also use \verb+\section*{...}+ to create an appendix without a letter attached. 
\end{document}